%% file: LIS.tex
\newif\ifOneColumn %one column format paper
\OneColumnfalse

\newif\ifLinksBlack %black colored links
\LinksBlackfalse
\ifOneColumn
\documentclass[journal,12pt,draftclsnofoot,onecolumn]{IEEEtran}
\else
\documentclass[journal]{IEEEtran}
\fi

\IEEEoverridecommandlockouts

\input{packages}

\ifLinksBlack
\RequirePackage[colorlinks,allcolors=black]{hyperref}
\else
\RequirePackage[colorlinks,allcolors=blue]{hyperref}
\fi

\graphicspath {{figures/}}

\input{commands}

\begin{document}
	
	\title{Distributed and Scalable Uplink Processing for LIS: Algorithm, Architecture, and Design Trade-offs}
	
	\author{
		Jes\'{u}s~Rodr\'{i}guez~S\'{a}nchez~\href{https://orcid.org/0000-0002-5531-1071}{\includegraphics[scale=0.04]{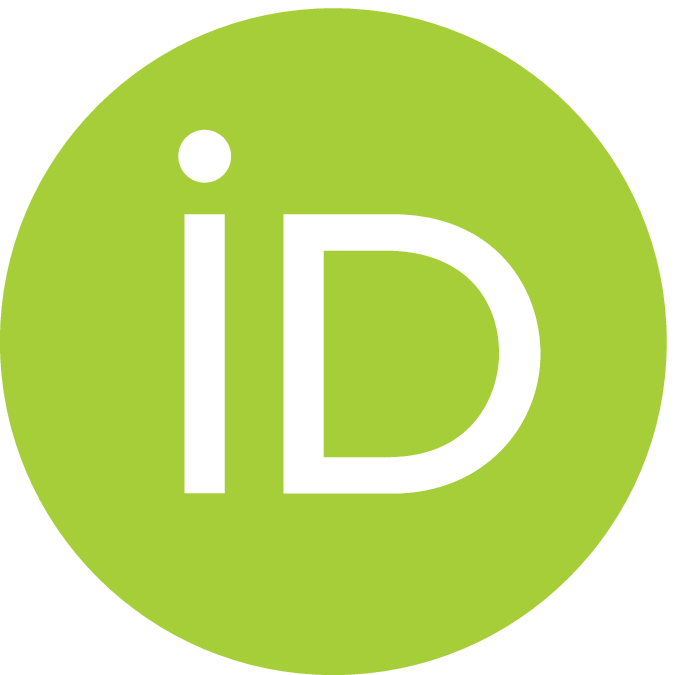}},~\IEEEmembership{Student~Member,~IEEE,}
		Fredrik~Rusek~\href{https://orcid.org/0000-0002-2077-3858}{\includegraphics[scale=0.04]{orcid.eps}},~\IEEEmembership{Member,~IEEE,}
		Ove~Edfors~\href{https://orcid.org/0000-0001-5966-8468}{\includegraphics[scale=0.04]{orcid.eps}},~\IEEEmembership{Senior~Member,~IEEE,}
		and~Liang~Liu~\href{https://orcid.org/0000-0001-9491-8821}{\includegraphics[scale=0.04]{orcid.eps}},~\IEEEmembership{Member,~IEEE}
		\thanks{This paper is build upon previous result presented at the 2020 IEEE ICC \cite{jesus_icc20}, and unpublished article with preliminary results \cite{jesus_iter}.
		The authors are with the Electrical and Information Technology, Lund University, 22363 Lund, Sweden (e-mail: \{jesus.rodriguez, fredrik.rusek, ove.edfors, liang.liu\}@eit.lth.se).}
	}
	
	\maketitle
	
	\begin{abstract}
		The Large Intelligent Surface (LIS) is a promising technology in the areas of wireless communication, remote sensing and positioning. It consists of a continuous radiating surface located in the proximity of the users, with the capability to communicate by transmission and reception (replacing base stations).
		Despite of its potential, there are numerous challenges from implementation point of view, being the interconnection data-rate, computational complexity, and storage the most relevant ones. In order to address those challenges, hierarchical architectures with distributed processing techniques are envisioned to to be relevant for this task, while ensuring scalability. In this work we perform algorithm-architecture codesign to propose two distributed interference cancellation algorithms, and a tree-based interconnection topology for uplink processing. We also analyze the performance, hardware requirements, and architecture trade-offs for a discrete LIS, in order to provide concrete case studies and guidelines for efficient implementation of LIS systems.
	\end{abstract}

	\section{Introduction}
	\label{section:intro}
	LIS has been identified as one of the key technologies for beyond 5G \cite{husha_data,husha_data2,husha_asign,husha_pos}. In Fig. \ref{fig:LIS_concept} we show the concept of a LIS serving multiple users simultaneously. The LIS is a continuous radiating surface located in the proximity of the users. Each part of the surface is capable of receiving and transmitting electromagnetic (EM) waves with a certain control, so the EM waves can be focused in 3D space with high resolution, opening the door of a new world of possibilities for power-efficient communication.
	
	Apart from LIS, another type of intelligent surface has been studied in the literature, which can be classified within the smart radio environment paradigm \cite{direnzo}, by which the wireless channel can be controlled to facilitate the transmission of information, as opposite to traditional wireless communication systems, where the channel is imposed by nature, and transmitter and receiver adapt to changes in it. One example of this new trend is the reconfigurable surfaces, known as \textit{intelligent reflecting surfaces}, \textit{programmable metasurfaces}, \textit{reconfigurable intelligent surfaces (RIS)}, and \textit{passive intelligent mirrors} among others \footnote[1]{We refer to \cite{basar} and \cite{huang_hol} for a complete list of surfaces.}, which consist of electronically passive surfaces with the capability to control how the waves are reflected when hitting their surface. Furthermore, the term LIS has also been recently used for such a passive surfaces \cite{taha,han,huang}, with the subsequent risk of confusion. While RIS can be seen as part of the radio channel, LIS acts as an active  basestation/access point. LIS contains full transmitters and receivers chains, together with baseband processing capabilities to transmit and receive. A list of the main differences between RIS and LIS is shown in Section \ref{sub:RIS}.

	\begin{figure}
		\centering
		\includegraphics[width=\linewidth]{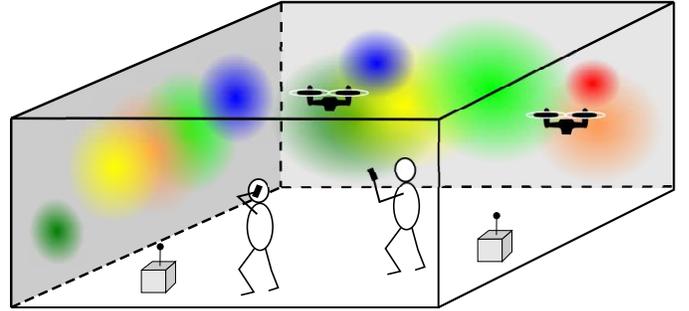}
		\caption{A LIS serving multiple users simultaneously.}
		\label{fig:LIS_concept}
		\vspace*{-4mm}
	\end{figure}

	Most of the research on LIS has been focusing on concept exploration  \cite{husha_data,husha_data2,husha_asign,husha_pos}, system performance \cite{jung,juan_CAMSAP}, and channel modeling \cite{dardari,williams}. However, the questions from implementation point of view have not yet been answered.  This paper aims to cover this area, by identifying and addressing implementation challenges, and providing design guidelines for an efficient implementation of LIS.
	
	The first step to make LIS implementable is to make it discrete (based on discrete antennas). It is known \cite{husha_data} that a continuous LIS can be replaced by a discrete one with no practical difference in achieved capacity. However, an efficient implementation of a discrete LIS is still very challenging, as it is expected to be made up of a very large number of antennas with the corresponding receiver (and transmitter) chains, which translates into a tremendous amount of inter-connection data-rate, that needs to be routed to the Central Digital Signal Processor (CDSP) through the backplane network. This centralized approach has already been employed in the LuMaMi Massive MIMO testbed \cite{LuMaMi}, with a need of $100$ bidireccional links, and a total aggregated interconnection bandwidth of 5GB/s. In case of LIS this number is much higher. To illustrate, let´s assume a $1.2m\times1.2m$ array containing $1,024$ antennas in the 4GHz band (assuming spacing of half wavelength), with the corresponding radio frequency (RF) and analog-to-digital converter (ADC) blocks. Then, if each ADC uses 12bits per I and Q, that makes a total rate of $\sim 48$Tb/s \footnote{Assuming 5G-NR standard, and sampling rate of $480,000 \cdot 4,096 \sim 2$Gs/s.}. This is $3$ orders of magnitude higher than the massive MIMO counterpart \cite{LuMaMi}, where this issue has been previously addressed \cite{cavallaro,puglielli,jesus_journal_MaMi,muris}. Therefore there is a need to come up with specific architectures and algorithms in order to overcome this bottleneck.
	
	We propose to tackle those challenges by algorithm and architecture co-design. At the algorithm level, we explore the unique features of LIS (e.g., very large aperture) to develop distributed algorithms that enable the processing being performed locally, near the antennas. This will significantly relax the requirement for interconnection bandwidth. At the hardware architecture design level, we propose to panelize the LIS in order to facilitate processing distribution, scalability, manufacturing, and installation. A hierarchical interconnection topology is developed accordingly to provide efficient and flexible data processing, and data exchange between panels and CDSP. Based on the proposed algorithm-architecture, extensive analysis has been performed to enable trade-offs between system capacity, interconnection bandwidth, computational complexity, and processing latency. This will provide high-level design guidelines for the real implementation of LIS systems. The contributions of this work are originated from our previous work in \cite{juan_VTC19,jesus_icc20,jesus_iter}, being considerably extended in the present paper.
	
	This article is organized as follows: Section \ref{section:LIS} introduces the LIS concept, then the system model is presented in Section \ref{section:sys_model}. Our proposed algorithms are described in Section \ref{section:algorithms}, and the architecture description in Section \ref{section:arch}. Analysis and design trade-offs are presented in Section \ref{section:analysis}, and finally conclusions in Section \ref{section:conclusions}.

	Notation: In this paper, lowercase, bold lowercase and upper bold face
	letters stand for scalar, column vector and matrix, respectively. The
	operations $(.)^T$, $(.)^*$ and $(.)^H$ denote transpose, conjugate and conjugate transpose respectively. $\IK$ represents the identity matrix of size $K \times K$. Operator $\diag(.)$ returns a block diagonal matrix built with the list of matrices in the argument.

	\section{Large Intelligent Surfaces}
	\label{section:LIS}

	\begin{figure*}[ht]
		\captionsetup[subfigure]{justification=centering}
		\centering
		\subfloat[LIS panel with 64 dual-port antennas]{
			\includegraphics[width=0.24\linewidth]{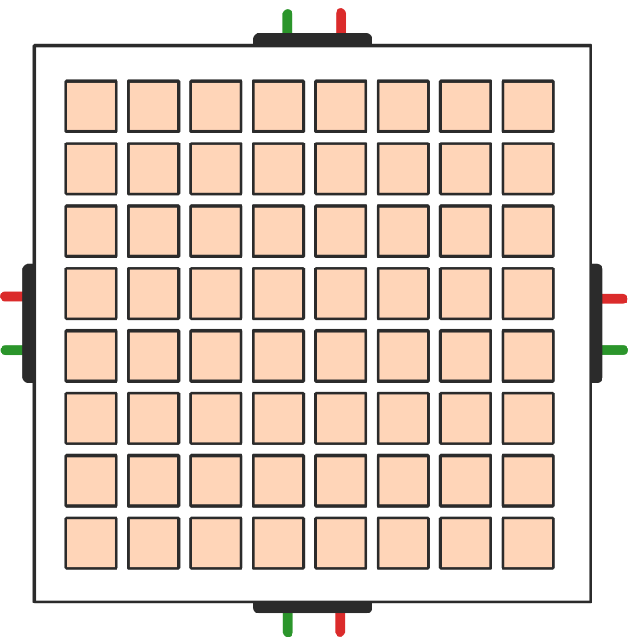}
			\label{fig:LIS_processing_dist_a}
		}
		\subfloat[Internal LIS panel architecture]{
			\includegraphics[width=0.24\linewidth]{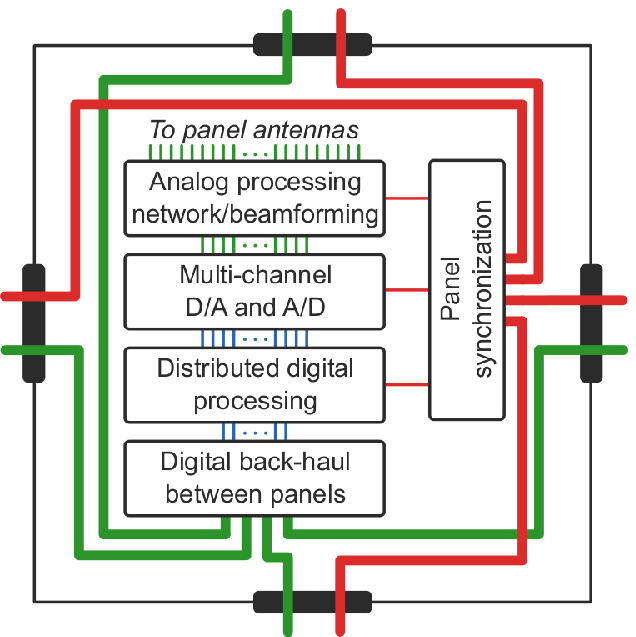}
			\label{fig:LIS_processing_dist_b}
		}
		\subfloat[Fully connected LIS using 16 panels]{
			\includegraphics[width=0.23\linewidth]{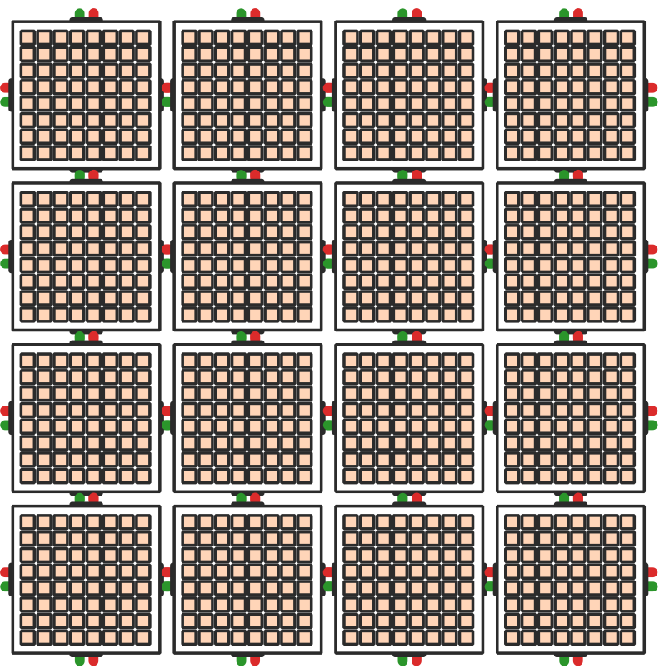}
			\label{fig:LIS_processing_dist_c}
		}
		\subfloat[Partially connected LIS using 6 panels]{
			\includegraphics[width=0.25\linewidth]{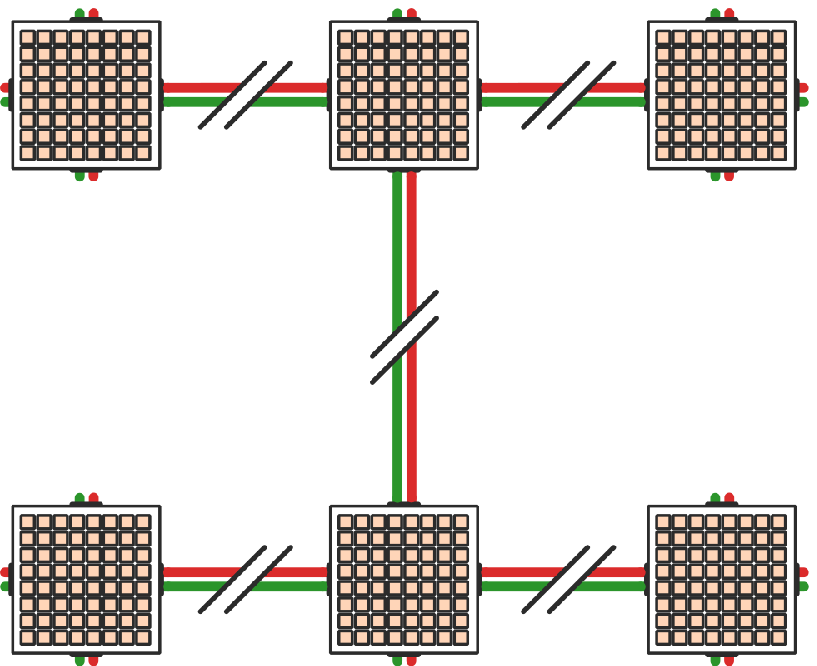}
			\label{fig:LIS_processing_dist_d}
		}
		
		\caption{LIS architecture components in the form of a) panel, b) each with internal analog and digital processing resources, synchronization, and digital back-haul. Identical panels can be combined in arbitrary configurations, e.g., fully or partially connected. Each panel contributes with its own processing resources, making the available resources for distributed processing fixed per area unit.}
		\label{fig:LIS_processing_dist}
		\vspace*{-5mm}
	\end{figure*}

	This section describes the key features of LIS, by comparing with massive MIMO and RIS. We also present the general concept of panelized LIS, which is proposed to ensure scalability and implementation feasibility.
	
	\subsection{Differences with Massive MIMO}
	\label{subsection:difference_mami}
	Multi-antenna technology has evolved in recent years in the form of Massive MIMO, where the number of antennas in the Base-Station (BS) grows up to $\sim 100$, bringing many benefits from communication and energy consumption points of view \cite{rusek}. LIS wants to go further by increasing the number of antennas one or two orders of magnitude more, altogether with the physical size of the array, which brings gains beyond Massive MIMO can provide. This results in fundamental differences between these two technologies, which are listed as follows:

	\begin{itemize}
		\item LIS aperture is larger in comparison to Massive MIMO, which translates into higher directivity and spatial multiplexing capabilities.
		\item Users are close to the LIS in relation to its size, what makes them being in the near field, as opposite to Massive MIMO (and others cellular access technologies) where users are in the (Fraunhofer) far field region. Being in the near field requires the use of channel models based on spherical waveforms, rather than the planar wave approximation, whose use is generalized in Massive MIMO (and other cellular technologies).
		\item Due to the lower path loss (due to the close proximity between users and LIS), and the large antenna gain, transmit power is expected to be relatively small for both sides of the communication, opening the door for extensive use of low-cost and low-power analog components.
		\item Received power distribution from users is not uniform throughout the surface as illustrated in Fig. \ref{fig:LIS_concept}. The same user is received with different signal intensity from different parts of the LIS. This can be exploited by the use of localized digital signal processing, leading to a more efficient use of computational resources, and inter-connection bandwidth, without significantly sacrificing the system performance. This is in contrast with  Massive MIMO (and other cellular technologies), where users are seen with same power across the antenna array (which is in fact connected to the planar wave approximation).
	\end{itemize}

	\subsection{Differences with RIS}
	\label{sub:RIS}
	As commented in the Introduction, LIS and RIS are fundamentally different technologies. In this section we summarize the main differences between these two:
	\begin{itemize}
		\item RIS acts as a programmable reflector between the radio access point and the users, forming part of the channel. Typically it is configured in a way to improve a certain quality metric, such as capacity. LIS acts as a radio access point capable to communicate directly to users.
		\item LIS contains full receivers (in contrast to most of RIS) and baseband processing capabilities to obtain CSI from pilots transmitted by users. This allows an accurate calculation of the corresponding equalization matrix, and further detection within LIS.
	\end{itemize}

	\subsection{Panelized Implementation of LIS}	
	Given that LIS is large in physical size and there is a need for distributed processing close to the antennas, we propose to divide the LIS in square units or panels. Panelization allows the LIS to adapt to a wide range of scenarios by adding, moving, or removing panels as desired, varying consequently the size and form of the LIS. Different shapes can be achieved by placing the panels in different ways: square, rectangular or distributed (panels not physically together, but covering a certain area). It also simplifies the system design, verification, and fabrication by only focusing on the panel as building block, instead of covering all possible LIS sizes and forms. Additionally, the installation becomes also simpler as the panel weights less, being easy to lift and mount.
	
	A high level overview of the LIS architecture components, processing distribution, and interconnection is shown in Fig. \ref{fig:LIS_processing_dist}. Panels are composed of a group of antennas forming a squared array as shown in Fig. \ref{fig:LIS_processing_dist_a}. Each panel contains internal processing resources in the analog and digital domains, and inter-connection capabilities to connect the panel to other panels (Fig. \ref{fig:LIS_processing_dist_b}). As said before, panels provide freedom to assembly the LIS. As an example, Fig. \ref{fig:LIS_processing_dist_c} shows 16 panels fully connected, forming a 1024-antennas LIS, while in Fig. \ref{fig:LIS_processing_dist_d}, 6 physically distant panels are connected in a distributed fashion (e.g: covering a certain volume in space, such as an office, or theater).

	\section{System Model}
	\label{section:sys_model}

	\begin{figure}[ht]
		\footnotesize
		\centering
		\psfrag{xh}{$\s$}
		\psfrag{N}{$N$}
		\psfrag{K}{$K$}
		\psfrag{P}{$P$}
		\psfrag{MP}{$M_{p}$}
		\psfrag{NP}{$N_{p}$}
		\psfrag{z}{$\z$}
		\psfrag{y}{$\y$}
		\psfrag{x}{$\x$}
		\psfrag{W1}{$\Wbf_{\mathrm{P},1}$}
		\psfrag{W2}{$\Wbf_{\mathrm{P},2}$}
		\psfrag{WP}{$\Wbf_{\mathrm{P},P}$}
		\psfrag{WB}{$\Wbf_{\mathrm{B}}$}
		\psfrag{H1}{$\Hbf_{1}$}
		\psfrag{H2}{$\Hbf_{2}$}
		\psfrag{HP}{$\Hbf_{P}$}
		\psfrag{Z1}{$\Z_{1}$}
		\psfrag{Z2}{$\Z_{2}$}
		\psfrag{ZP1}{$\Z_{P-1}$}
		\psfrag{ZP}{$\Z_{P}$}
		\psfrag{TP1}{$\text{To Panel 1 (optional)}$}
		\psfrag{FE}{$\text{Front-End}$}
		\psfrag{BP}{$\text{Backplane}$}
		\psfrag{local}{$\text{local}$}
		\includegraphics[width=0.95\linewidth]{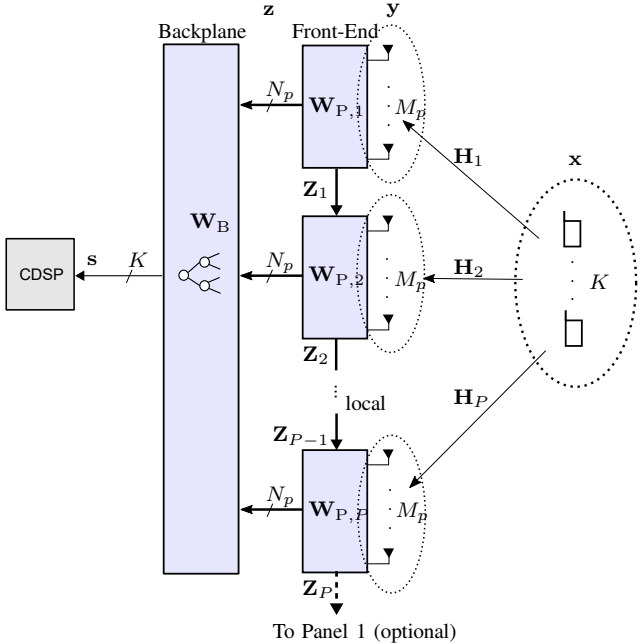}
		\caption{$K$ users transmitting to an M-elements discrete-LIS formed by $P$ panels.}
		\label{fig:system_model}
	\end{figure}
	
	A conceptual view of a discrete LIS system is presented in Fig. \ref{fig:system_model}. We consider $K$ users transmitting to the LIS, which is divided in three parts: \textit{front-end}, \textit{backplane}, and CDSP. We will use the term \textit{front-end} to refer to the per-antenna processing which is performed locally at each panel, and \textit{backplane} to the related processing involving data aggregation, distribution, and processing for further dimensionality reduction. Backplane can be made of multiple levels and processing nodes as we will present in Section \ref{section:arch}. The processing unit in the front-end is the Local DSP (LDSP), while the one in backplane is the Backplane DSP (BDSP). The data is finally collected by the CDSP for detection. In the present section we also introduce a mathematical model for the communication and the LIS-baseband processing.
	
	We consider the transmission from $K$ single antenna users to the  LIS containing $M$ active antenna elements (input dimensionality). The LIS is divided into $P$ squared panels, each with $M_{\text{p}}$ elements, such that $M_{\text{p}} \cdot P = M$. Each panel has an output with $N_p$ dimensions, and the total number of them is $N$, such that $N = N_{\text{p}} \cdot P$. Panels are connected to the backplane, which collects their output data, process it, and provides the CDSP with $K$ values to ensure proper detection. The data dimensionality is reduced from the antenna elements interface (vector $\y \in \mathbb{C}^{M}$ in the figure) to the backplane input ($\z \in \mathbb{C}^{N}$) due to the front-end, and from this to the CDSP interface ($\s \in \mathbb{C}^{K}$) due to backplane processing. We assume $M \gg K$ for the rest of the article.
	
	The $M\times 1$ received vector at the LIS is given by
	\begin{equation}
	\mathbf{y} = \sqrt{\rho}\Hbf\mathbf{x}+\mathbf{n},
	\label{eq:received_signal}
	\end{equation}
	where $\x$ is the transmitted $K\times 1$ user data vector, and $\E\{\x \x^{H}\}=\IK$. $\Hbf$ is the  channel matrix, and $\mathbf{n} \sim \mathcal{CN}(0,\I)$ is a $M \times 1$ noise vector, that we assume with identity covariance for simplicity without loss of generality. This convention leaves $\rho$ as the "transmit" SNR and therefore it is dimensionless.
	
	Assuming the location of user $k$ is $(x_{k},y_{k},z_{k})$, where the LIS is at $z=0$. The channel between this user and a LIS antenna at location $(x,y,0)$ is given by the complex value \cite{husha_data}
	\begin{equation}
	h_{k}(x,y)=\frac{\sqrt{z_{k}}}{2\sqrt{\pi} d_{k}^{3/2}}\exp{\left( -\frac{2\pi j d_{k}}{\lambda} \right)},
	\label{eq:channel}
	\end{equation}
	where $d_{k}=\sqrt{z_{k}^{2}+(x_{k}-x)^2+(y_{k}-y)}$ is the distance between the user and the antenna, and Line of Sight (LoS) propagation between them is assumed. $\lambda$ is the wavelength. The channel matrix can be expressed as
	\begin{equation}
	\Hbf=[\Hbf_{1}^{T},\Hbf_{2}^{T},\cdots \Hbf_{P}^{T}]^{T},
	\label{eq:H_structure}
	\end{equation}
	where $\Hbf_{i}$ is the $\Mp \times K$ channel matrix of the $i$-th panel.
	We assume each panel has perfect knowledge of its local channel.

	\subsection{Dimensionality reduction: A lossless or lossy process}
	\label{sub:lossles_vs_lossy}
	As commented previously, our LIS architecture can be seen as a system to reduce the dimensionality of the very large incoming signal ($M \times 1$) down to a value required for detection at the CDSP ($K \times 1$). We can classify this process attending to the criteria of preserving information as: lossless and lossy. A lossless process maintains the mutual information between CDSP input and user's data, formally
	\begin{equation}
	I(\s;\x) = I(\y;\x),
	\nonumber
	\end{equation}
	so the system can achieve channel capacity performance if optimal processing is done in CDSP. Initial progress on the trade-offs of distributed processing for MIMO systems in lossless approach can be seen in \cite{juan_icc20}, and more recently in \cite{juan_journal}. In this regime $\Np \geq \min\{\Mp,K\}$.
	
	In despite of the attractiveness of achieving optimal performance, lossless presents a high cost from implementation point of view, as it requires larger panel output dimensionality, which translates in higher interconnection bandwidth throughout the backplane. In this article we look for a good compromise between implementation cost and performance, which leads us to explore the case $\Np \leq \Mp$ \footnote{We note that this is equivalent to: $N \leq M$.}, and especially $\Np \ll \Mp$. By selecting this regime we expect to reduce significantly interconnection bandwidth at the cost of a loss in performance, which can be expressed formally as
	\begin{equation}
	I(\s;\x) \leq I(\y;\x).
	\nonumber
	\end{equation}

	Our approach is to include enough flexibility into the system to obtain enough working points to establish a rich trade between implementation cost and performance, which in fact, allows the system to adapt to a large variety of scenarios during the deployment phase. As we will see in Section \ref{section:analysis}, it is possible achieve close to channel capacity conditions with significant reduction in implementation cost.

	\subsubsection{Filtering}
	\label{sub:filtering}
	In order to achieve dimensionality reduction, we employ linear filtering in the incoming data, while to achieve enough flexibility we consider separate filters for front-end and backplane.
	
	Let us consider the panelized architecture shown in Fig. \ref{fig:system_model}, where each panel performs local per-antenna processing on the received signal and delivers the result to the backplane. There is not cooperation among panels during front-end filtering, therefore the filter matrix $\Wp$ has the following structure
	\begin{equation}
	\Wp = \diag(\Wpi{1},\Wpi{2},\cdots,\Wpi{P})
	\label{eq:W_structure}
	\end{equation}
	where $\Wpi{i}$ is the $\Mp \times \Np$ matrix filter of the $i$-th panel.
	
	Then the front-end output is given by
	\begin{equation}
	\z= \Wph \y = \sqrt{\rho}\Wph\Hbf \mathbf{x} + \hatn,
	\label{eq:filtering_fe}
	\end{equation}
	where $\hatn=\Wph\mathbf{n}$ is the filtered noise. Mind that size of $\z$ is $N$, and $N \leq M$ according to the reasoning in this section. Finally, the backplane filters $\z$ in order to obtain $\s$ as
	\begin{equation}
	\s= \Wbf_{\mathrm{B}}^{H} \z,
	\label{eq:filtering_bp}
	\end{equation}
	which is used by CDSP for detection.
	
	\subsection{Sum-Rate Capacity}
	\label{sub:sum-rate}
	The mutual information between $\z$ and $\x$ is $I(\x;\z)=H(\z)-H(\z|\x)$. Assuming white Gaussian signaling transmitted by users, the mutual information for a given $\Hbf$ and $\Wp$ can be further expanded as
	\begin{equation}
	\begin{split}
	I(\x;\z) &= \log_{2}|\Sbf_{\z\z}| - \log_{2}|\Sbf_{\hatn\hatn}|\\
	&= \log_{2} |\rho\Wph\Hbf\Hbfh\Wp+\Wph\Wp|\\
	&-\log_{2} |\Wph\Wp|,
	\end{split}
	\label{eq:MI}
	\end{equation}
	where $\Sbf_{\z\z}$ and $\Sbf_{\hatn\hatn}$ are the covariance of the multivariate complex gaussian vector $\z$ and $\hatn$ respectively.
	If $\Wp$ is full-rank matrix, and taking into account that $M \geq N$, then $(\Wph \Wp)^{-1}$ exists and we can rewrite \eqref{eq:MI} as
	\begin{equation}
	\begin{split}
	I(\x;\z) &= \log_{2} |\IK + \rho\Hbfh\Wp(\Wph\Wp)^{-1}\Wph\Hbf|.\\
	\end{split}
	\label{eq:MI2}
	\end{equation}

	We are interested in maximize the sum-rate capacity for this front-end architecture, and it will be the maximum of \eqref{eq:MI2} over all possible $\Wp$ for a given $\Hbf$. If we take into account the block structure of $\Hbf$ and $\Wp$ presented in \eqref{eq:H_structure} and \eqref{eq:W_structure} respectively, the sum-rate capacity at $\z$ interface is given by
	\begin{equation}
	\begin{split}
	C_{\z} &= \max_{\{\Wpi{i}\}} \log_{2} | \IK + \rho \sum_{i=1}^{P} \Hbf_{i}^{H} \Wpi{i} (\Wphi{i} \Wpi{i})^{-1} \Wphi{i} \Hbf_{i} |\\
	&= \max_{\{\Q_{i}:  \Q_{i}^{H}\Q_{i}=\I_{\Np}\}} \log_{2} | \IK + \rho \sum_{i=1}^{P} \Hbf_{i}^{H} \Q_{i} \Q^{H}_{i} \Hbf_{i} |,\\
	\end{split}
	\label{eq:Cz_multi}
	\end{equation}
	where $\Q_{i}$ is a $\Mp \times \Np$ semi-unitary matrix, consisting of the $\Np$-first singular vectors of $\Wpi{i}$.
	For the last expression in \eqref{eq:Cz_multi}, it is assumed that all matrices $\Wphi{i}\Wpi{i}$ are full-rank, so the inverse exists.
	
	As we will show in next section, selection of $\{\Wpi{i}\}$ is done in a way that each element is semi-unitary, which leads to white noise at the front-end output. Therefore, once the front-end filters are selected, they can be seen as part of the channel by the backplane, and we can apply same reasoning to obtain $\Wbf_{\mathrm{B}}$, leading to \footnote{A detailed explanation of this process can be found in Section \ref{section:arch}.}
	\begin{equation}
	C_{\s} = \max_{\Wb} \log_{2} | \IK + \rho \widetilde{\Hbf}^{H} \Wb (\Wb^{H} \Wb)^{-1} \Wb \widetilde{\Hbf} |,
	\label{eq:Cs}
	\end{equation}
	where $\widetilde{\Hbf}=\Wph \Hbf$ is the equivalent channel.

	\section{Distributed algorithms for dimensionality reduction}
	\label{section:algorithms}
	In this section we introduce two algorithms to obtain the filtering matrices $\{\Wpi{i}\}$ and $\Wb$, which are executed in the LDSP and BDSP respectively. The way the algorithms are explained here refers to the panels for simplicity, but can be extended to the backplane by using as channel matrix the equivalent one $\widetilde{\Hbf}$, presented before, and $P$ equal to the number of processing nodes in backplane. More details about the backplane case can be found in Section \ref{section:arch}.
	
	The first of the algorithms is a straightforward approach with relatively low computational complexity based on the known Matched Filter (MF) method, which we select conveniently as a comparison baseline for our proposed algorithm.
	
	\subsection{Reduced Matched Filter (RMF)}
	RMF consists of a reduced version of the known MF method. In this case, the filter  $\Wbf_{i}$ is built by the $N_{\text{p}}$ $\textit{strongest}$ columns of $\Hbf_{i}$. The $\textit{strenght}$ of a column $\h_{n}$ is defined as $\|\h_{n}\|^{2}$. The $\Mp \times \Np$ filtering matrix of the $i$-th panel is then expressed as
	\begin{equation}
	\Wbf_{\text{RMF},i} = \left[ \h_{k_1}, \h_{k_2}, ...,  \h_{k_{Np}} \right],
	\label{eq:W_RMF}
	\end{equation}
	where $\h_{n}$ is the $\Mp \times 1$ channel vector for the $n$-th user, and $\{k_{i}\}$ the set of indexes relative to the $N_{\text{p}}$ strongest users \footnote{This is connected to the non-uniform user power distribution in the LIS, described in Section \ref{subsection:difference_mami}, which translates to the fact that a panel may not see all users with same power, which depends on their physical proximity.}.
	
	When RMF is applied at the panel level as local filtering, each output is associated to a certain user. Therefore, nodes in the backplane can combine data coming from the same user, in a similar fashion as in distributed MF \cite{jeon_li2}. The result of the filtering is available at CDSP input for final detection (hard or soft). It is important to notice that in this method front-end processing nodes can work independently, without sharing channel related information. This saving in interconnection bandwidth comes with a performance loss as we will see in Section \ref{section:analysis}.
	
	\subsection{Iterative Interference Cancellation (IIC)}
	The IIC algorithm aims to solve the optimization problem described in \eqref{eq:Cz_multi}. It is an iterative algorithm based on a variant of the known multiuser water-filling method \cite{Yu}. The pseudocode is shown in Algorithm \ref{algo:MUWF}.
	\IncMargin{1em}
	\begin{algorithm}[ht]
		\SetKwInOut{Input}{Input}
		\SetKwInOut{Output}{Output}
		\SetKwInOut{Preprocessing}{Preprocessing}
		\SetKwInOut{Init}{Init}
		\Input{$\{\Hbf_{i}\}, i=1 \cdots P$}
		\Preprocessing{ $\Q_{i} = \mathbf{0}, i=1 \cdots P$}
		\Repeat {sum-rate converges} {
			\For{$i = 1,2,...,P$}{
				$\Z_{i} = \IK + \rho \sum_{j=1,j\neq i}^{P} \Hbf_{j}^{H} \Q_{j} \Q^{H}_{j} \Hbf_{j}$\\
				
				$\mathbf{\Q}_{i} = 
				\argmax_{\overline{\Q}_i}
				|\rho \Hbf_{i}^{H} \overline{\Q}_i \overline{\Q}_i^{H} \Hbf_{i} + \Z_{i}|$\\
				$\text{subject to } \overline{\Q}_i^H \overline{\Q}_i = \I_{\Np}$
				
			}
		}
		\Output{$\{\Q_{i}\}, i=1 \cdots P$}	
		\caption{IIC algorithm pseudocode}
		\label{algo:MUWF}
	\end{algorithm}\DecMargin{1em}
	The algorithm splits the joint optimization problem \eqref{eq:Cz_multi} into $P$ small ones, which are solved in a sequential basis.
	The goal of the algorithm is to calculate the $\Mp \times \Np$ matrices $\{\Q_{i}\}$. The product $\Q_{i}\Q^{H}_{i}$ is low-rank as $\Np \leq \Mp$, which exploits the fact that only a few users are conveniently seen by each panel (ideally this number is $\Np$). The fundamental difference between our current algorithm and \cite{Yu} is due to the low-rank constraint present in our proposed algorithm.
	
	At each iteration of the algorithm, $K \times K$ matrix $\Z_i$ is obtained as intermediate result, which contains contribution from the rest of panels, and plays the role of noise covariance in the sum-rate optimization problem formulated in line 4. The algorithm iterates over all panel indexes, as many times as needed until a certain convergence criteria is achieved.

	\subsection{Processing Distribution}
	\label{sub:dist}
	
	It is natural to map each iteration of the IIC algorithm to each panel, as it requires local CSI, while $\Z_{i}$ can be computed also locally as an update of $\Z_{i-1}$. Therefore, each panel computes and shares $\Z_{i}$ with the neighbor panel, $i+1$, while $\Q_{i}$ is stored locally for further filter calculation, and not shared.
	
	We propose that panels are connected by fast local and dedicated connections for the exchange of data related to matrix $\Z$. In general, we can say that "the matrix $\Z$ is passed from panel to panel" using the dedicated connections depicted in Fig. \ref{fig:system_model}. This decentralized approach is described in Algorithm \ref{algo:IIC_i} for a certain panel $i$ \footnote{For simplicity and to limit latency, we consider only one iteration to the set of panels throughout the rest of this article. We are aware that increasing the number of iterations improves the performance.}.

	\IncMargin{1em}
	\begin{algorithm}[ht]
		\SetKwInOut{Input}{Input}
		\SetKwInOut{Output}{Output}
		\SetKwInOut{Preprocessing}{Preprocessing}
		\SetKwInOut{Init}{Init}
		\Preprocessing{ $\mathbf{Z}_{0} = \IK$}
		\Input{$\{\Hbf_{i}, \Z_{i-1}\}$}
		$\Q_{i} = \argmax_{\overline{\Q}_i} |\rho \Hbf_{i}^{H} \overline{\Q}_{i} \overline{\Q}_{i}^{H} \Hbf_{i} + \mathbf{Z}_{i-1}|$\\
		$\text{subject to } \overline{\Q}_i^H \overline{\Q}_i = \I_{\Np}$\\
		$\mathbf{Z}_{i} = \mathbf{Z}_{i-1} + \rho\Hbf_{i}^{H} \Q_{i} \Q^{H}_{i} \Hbf_{i}$
		\caption{Decentralized IIC algorithm at $i$-th panel}
		\label{algo:IIC_i}
		\Output{$\{\Q_{i}, \mathbf{Z}_{i}\}$}
	\end{algorithm}\DecMargin{1em}

	The solution to the local optimization problem at $i$-th panel is $\Q_{i}=[\hatu_{1},\hatu_{2}, \cdots, \hatu_{\Np}]$, where $\hatu_{n}$ is the $n$-th left-singular vector of $\hat{\Hbf}_{i}=\Hbf_{i} \U_{z} \Sbf_{z}^{-1/2}$ corresponding to the n-$th$ ordered singular value, and $\mathbf{Z}_{i-1} = \U_{z} \Sbf_{z} \U_{z}^{H}$ the eigen-decomposition of $\Z_{i-1}$ (see Appendix-\ref{proof:sol_IIC_i} for proof).
	
	The pseudocode for the processing at the $i$-th panel is shown in Algorithm \ref{algo:IIC},
	\IncMargin{1em}
	\begin{algorithm}[ht]
		\SetKwInOut{Input}{Input}
		\SetKwInOut{Output}{Output}
		\SetKwInOut{Preprocessing}{Preprocessing}
		\SetKwInOut{Init}{Init}
		\Input{$\{\Hbf_{i}, \Z_{i-1}\}$}
		$[\U_z,\Sbf_z] = \text{svd} (\Z_{i-1})$\\
		$\widetilde{\Hbf}_{i}=\Hbf_{i} \U_{z} \Sbf_{z}^{-1/2}$\\
		$\widetilde{\U} = \text{svd} (\widetilde{\Hbf}_{i})$\\
		$\Q_{i} = \widetilde{\U}(:,1:\Np)$\\
		$\Z_i = \Z_{i-1} + \rho\Hbfh_i \Q_i \Q^{H}_i \Hbf_{i}$
		\caption{Decentralized IIC algorithm processing steps for $i$-th panel}
		\label{algo:IIC}
		\Output{$\{\Q_i, \Z_i\}$}
	\end{algorithm}\DecMargin{1em}
	where $\widetilde{\U}$ is the left unitary matrix of $\widetilde{\Hbf}$, and $\Q_i$ is made by the eigenvectors associated to the  $\Np$ strongest singular values.
	
	\subsection{Selection of $\Wbf$ in IIC Algorithm}
	\label{sub:selection_W}
	In the single panel case (centralized LIS), the optimal selection of $\Q$ leads to $\Q = \Htilde^{H}$, where $\Htilde$ is a $M \times N$ semi-unitary matrix made by the $N$-first left singular vectors of $\Hbf$. Then, capacity will be given by the first $N$ largest singular values of $\Hbf$. Once $\Q$ is known, in order to select $\Wbf$, we notice that $\Wbf = \Q \widetilde{\Sbf}_{W} \V_{W}^{H}$, where $\widetilde{\Sbf}_{W}$ is a diagonal $N \times N$ matrix containing the $N$ largest singular values of $\Wbf$.
	Selection of $\widetilde{\Sbf}_{W}$ and $\V_{W}$ does not play any role in the sum-rate capacity, but the right choice can provide some benefits in other areas. In this work we choose $\widetilde{\Sbf}_{W} = \I_{N}$ to make $\Wbf$ semiunitary matrix, which brings a benefit in terms of reduction of interconnection bandwidth, that will be explained in next section. Selection of $\V_{W}$ can be arbitrary, and for simplicity we choose $\V_{W}=\I_{N}$. However, other unitary matrices are also valid, and could offer some advantages, but we do not cover this in the present work.
	
	In the multiple panel case, \eqref{eq:Cz_multi} represents a joint optimization problem among the matrices in the set $\{\Q_{i}\}$.  Similarly to the single panel case, $\Wbf_{i} = \Q_{i} \widetilde{\Sbf}_{W,i} \V_{W,i}^{H}$. Therefore, once $\Q_{i}$ is obtained, the selection of $\widetilde{\Sbf}_{W,i}$ and  $\V_{W,i}^{H}$ will follow identical considerations, this is: $\widetilde{\Sbf}_{W,i} = \I_{\Np}$, and $\V_{W,i}^{H} = \I_{\Np}$.
	
	\section{Interconnection Topology and DSP architecture}
	\label{section:arch}

	In this section we describe the proposed LIS architecture, including interconnection topology, and LDSP internal architecture able to support both RMF and IIC algorithms.
  
    \subsection{Tree-based Global Interconnection and Processing}
	\label{sub:backplane}

    \begin{figure}[ht]
    	\footnotesize
    	\centering
    	\psfrag{W1}{$\Wpi{1}$}
    	\psfrag{W4}{$\Wpi{4}$}
    	\psfrag{W61}{$\Wpi{61}$}
    	\psfrag{W64}{$\Wpi{64}$}
    	\psfrag{W1_1}{$\Wbij{1}{1}$}
    	\psfrag{W4_1}{$\Wbij{4}{1}$}
    	\psfrag{W13_1}{$\Wbij{13}{1}$}
    	\psfrag{W16_1}{$\Wbij{16}{1}$}
    	\psfrag{W1_2}{$\Wbij{1}{2}$}
    	\psfrag{W4_2}{$\Wbij{4}{2}$}
    	\psfrag{W1_3}{$\Wbij{1}{3}$}
    	\psfrag{MP}{\tiny $\Mp$}
    	\psfrag{NP}{$\Np$}
    	\psfrag{NB1}{$\Nb^{(1)}$}
    	\psfrag{4NB1}{$4\Nb^{(1)}$}
    	\psfrag{NB2}{$\Nb^{(2)}$}
    	\psfrag{NB3}{$\Nb^{(3)}$}
    	\psfrag{CDSP}{\tiny $\text{CDSP}$}
    	\psfrag{BDSP}{$\text{BDSP}$}
    	\psfrag{Wh}{$\Wbfh$}
    	\psfrag{WhH}{\color{blue} $\Wbfh\Hbf$}
    	\psfrag{Why}{\color{red} $\Wbfh\y$}
    	\psfrag{H}{\color{blue} $\Hbf$}
    	\psfrag{W}{\color{blue} $\Wbf$}
    	\psfrag{IIC}{$\mathrm{IIC}$}
    	\psfrag{RMF}{$\mathrm{RMF}$}
    	\psfrag{y}{\color{red} $\y$}
    	\psfrag{BP}{$\text{Backplane}$}
    	\psfrag{FE}{$\text{Front-End}$}
    	\includegraphics[width=\linewidth]{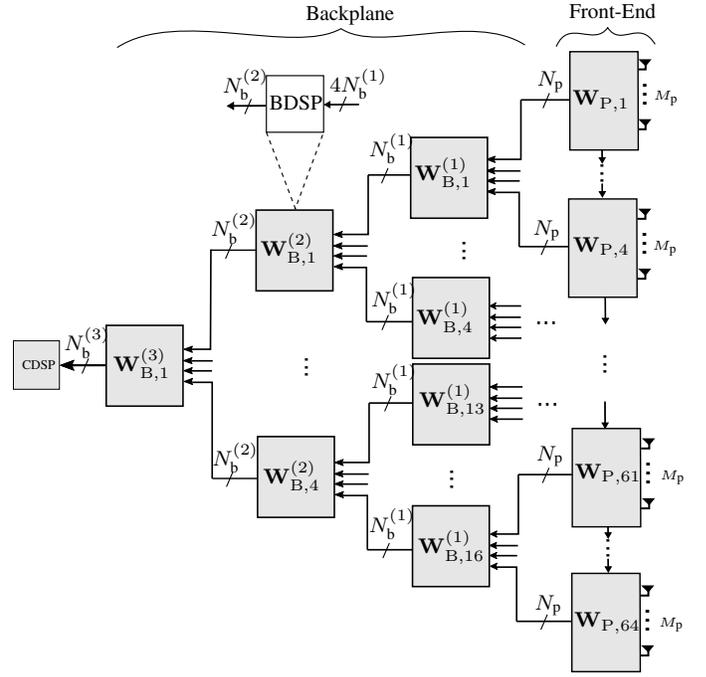}
    	\caption{Front-end and backplane tree topology and interconnection for a 64-panels LIS. Each panel contains a LDSP for distributed MIMO processing. Additionally, each node in the tree contains a BDSP unit, which aggregates data from 4 nodes, process, and delivers the result to the next node after corresponding dimensionality reduction, this is: $\Nb^{(i+1)} \leq 4 \Nb^{(j)},  i=1,2$, and $\Nb^{(1)} \leq 4 \Np$.}
    	\label{fig:model_backplane}
    	\vspace*{-5mm}
    \end{figure}

    In order to further increase the dimensionality reduction of the incoming data, while performing spatially local processing, we propose a hierarchical interconnection based on tree topology. The tree represents a distributed backplane, where front-end processing nodes are the leaves, and their outputs are combined in backplane nodes through multiple levels, reducing the total inter-connection bandwidth each time, until the resulting data is delivered to the CDSP. This process if shown in Fig. \ref{fig:model_backplane}. The main idea is to enable system scalability by adding levels in the tree as the LIS grows (more panels), while keeping the CDSP resources demand constant (dependent only on K) regardless of the LIS size. Another benefit of the tree topology is its low latency (the latency grows logarithmically with the number of panels).
    
   	As shown in the figure, the LIS backplane constitutes a 4-ary tree, which acts as an adaptation between the panels and the CDSP, introducing an extra dimensionality reduction of the incoming signal down to a level which can be efficiently transfered and handled by the CDSP, but high enough to allow good detection performance. Each node in the backplane contributes to $\Wb$, and aggregates data from 4 nodes, process it and deliver the output to the next node. The dimensionality of the output is lower or equal to the input, this is: $\Nb^{(i+1)} \leq 4 \Nb^{(j)},  i=1,2$, and $\Nb^{(1)} \leq 4 \Np$. This reduction is accumulated for the different consecutive levels of nodes the signal goes through.
	
	Let as assume the panels, during the formulation phase and after they obtain their local filtering matrix $\Wpi{i}$ (according to the selected algorithm), deliver the product $\Wpi{i}^{H} \Hbf_{i}$ ($\Np \times K$) to the corresponding node in the backplane. This can be seen as the result of filtering over the incoming pilot signals, which requires same amount of data as the filtering phase does. This product is the equivalent channel between the panel output and the users. A node aggregating outputs from 4 panels ($4 N_{p}$) can see those incoming values as an equivalent channel including the wireless channel and the 4 panels combined. The dimensionality of this equivalent channel is $4 N_{p}$, which is lower compared to the $4 M_{p}$ at the antenna level, but we expect it carries most of the captured channel capacity. If we take into account the selection of $\Wbf_{i}$ in the panels as semiunitary matrices according to Subsection \ref{sub:selection_W}, then the noise will be also white at the panel output. And filtered noise from 4 adjacent panels is still white due to the independence property of noise of different antennas/panels. Therefore at any node in the backplane connected to the panels we have same model as in \eqref{eq:received_signal} with the equivalent channel instead of $\Hbf_{i}$, and the filtered noise instead of $\mathbf{n}$, but with same covariance (identity matrix) \footnote{In case of not using semiunitary matrices, the noise gets colored and the covariance needs to be taken into account for sum-rate capacity optimization, therefore this noise covariance matrix needs also to be transfered between nodes in the tree. Selecting semiunitary matrices for the filters saves from this requirement.}. See Appendix-\ref{proof:white_noise} for proof. This means that \eqref{eq:filtering_fe} and \eqref{eq:filtering_bp}, and the sum-rate capacity derivation is also valid in this case with the equivalent channel, and the filter $\Wbf^{(1)}_{i}$ can be found by solving the optimization problem \eqref{eq:Cz_multi}. To obtain the filtering matrices, we follow the same problem described in Section \ref{sub:sum-rate}, with the same considerations as in \ref{sub:selection_W} for $\Wbf_{i}$ selection. Because we have the same problem for dimensionality reduction as in the front-end, both algorithms described in Section \ref{section:algorithms} can also be used in this case. This process can be repeated recursively for all levels of the tree up to the CDSP, which receives the total equivalent $K \times K$ channel matrix between the CDSP input interface and the users. This is used by the CDSP for detection. The general formulation algorithm to be executed at a certain LDSP or BDSP follows the steps shown in Algorithm \ref{algo:gen_form_alg}, where $\Heq$ is the equivalent channel matrix from current node input interface to users \footnote{Our experimental results shows no performance improvement by sharing $\Z$ among backplane nodes. Due to this reason we skip its use in Figure \ref{fig:model_backplane}}.

	\begin{algorithm}
		\DontPrintSemicolon
		\SetAlgoLined
		\SetKwInOut{Input}{Input}\SetKwInOut{Output}{Output}
		\Input{$\{\Heq,\Z\}$}
		
		\BlankLine

		\eIf{algorithm == IIC}{
			$\Wbf = \text{IIC}(\Heq, \Z)$\;
		}{
		$\Wbf = \text{RMF}(\Heq)$\;
		}
		\Output{$\{\Wbfh\Heq,\Z\}$}
		\caption{General formulation algorithm for tree-based LIS}
		\label{algo:gen_form_alg}
	\end{algorithm}
	\vspace*{-5mm}

	\subsection{DSP in panel and backplane nodes}

	\begin{figure*}[ht]
	 	\centering
	 	\psfrag{Mp}{$M_{\mathrm{p}}$}
	 	\psfrag{Np}{$N_{\mathrm{p}}$}
	 	\psfrag{ADC}{$\tiny\text{ADC}$}
	 	\psfrag{RF}{$\tiny\text{RF}$}
	 	\psfrag{ctrl}{$\text{ctrl}$}
	 	\psfrag{CDSP}{$\small\text{CDSP}$}
	 	\psfrag{LDSP}{$\text{LDSP}$}
	 	\psfrag{RCDSP}{$R_\text{CDSP}$}
	 	\psfrag{SPU}{$\text{SPU}$}
	 	\psfrag{Wh}{$\Wbfh_{i}$}
	 	\psfrag{WhH}{\color{blue} $\Wbfh_{i}\Hbf_{i}$}
	 	\psfrag{Why}{\color{red} $\Wbfh_{i}\y_{\mathrm{p}}$}
	 	\psfrag{H}{\color{blue} $\Hbf_{i}$}
	 	\psfrag{W}{\color{blue} $\Wbf_{i}$}
	 	\psfrag{FU}{$\text{FU}$}
	 	\psfrag{ALG}{$\tiny \text{(IIC/RMF)}$}
	 	\psfrag{MEM}{$\tiny \text{MEM}$}
	 	\psfrag{CE}{$\text{CE}$}
	 	\psfrag{FFT}{$\text{FFT}$}
	 	\psfrag{y}{\color{red} $\y_{\mathrm{p}}$}
	 	\psfrag{SPU}{$\text{SPU}$}
	 	\psfrag{PT}{$\text{Processing Tree}$}
	 	\psfrag{PP}{$P$}
	 	\psfrag{ppm1}{$\text{panel p-1}$}
	 	\psfrag{pp}{$\text{panel p}$}
	 	\psfrag{ppp1}{$\text{panel p+1}$}
	 	\psfrag{Zp}{$Z_\text{p}$}
	 	\psfrag{Zpm1}{$Z_\text{p-1}$}
	 	\includegraphics[width=0.9\linewidth]{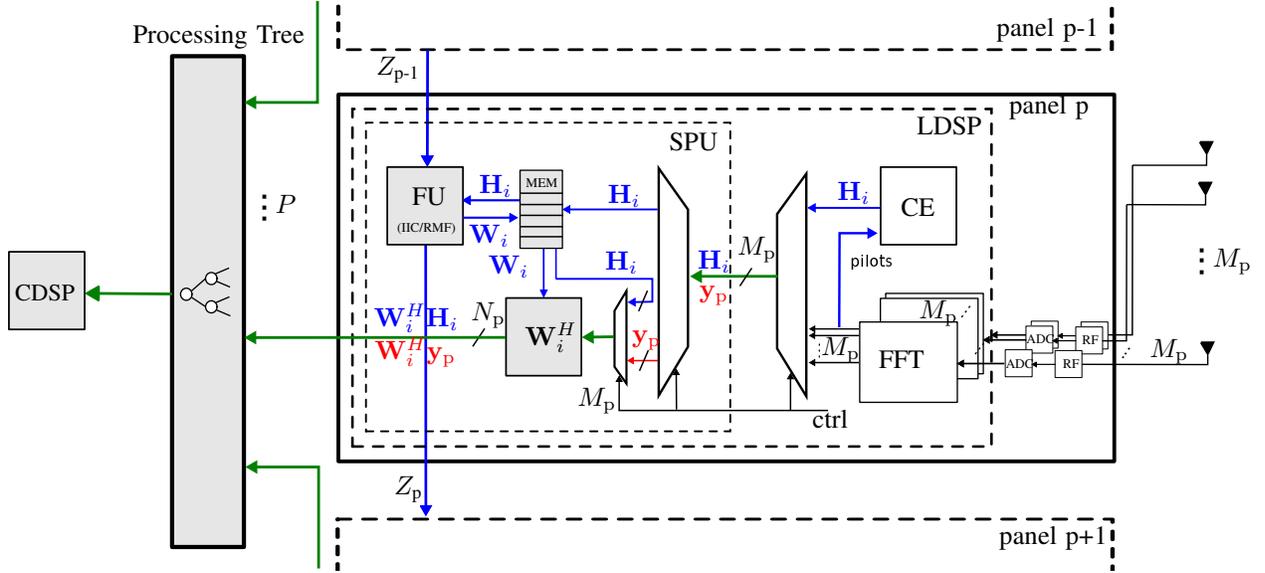}
	 	\caption{Overview of the Local DSP and Spatial Processing Unit (SPU) in a panel. Panel-panel, and panel-backplane connections are also shown. Blue lines are used only in formulation phase. Blue letters relate to data which is generated/transfered during formulation. Red ones refer to filtering phase. Green lines are used in both phases. In those cases, blue and red data structures are shown above and below the line. ctrl represents control line to switch between formulation and filtering phases.}
	 	\label{fig:panel}
	 	\vspace*{-3mm}
	\end{figure*}

	The internal architecture of the panel together with LDSP is depicted in Fig. \ref{fig:panel}. LDSP comprises all digital signal processing involved in the uplink tasks. After the RF and ADC, digitalized incoming signal is processed by FFT blocks to perform time-to-frequency domain transformation. During the formulation phase, the Channel Estimation block (CE) estimates a new $\Hbf_i$ for each channel coherence interval. In this paper we assume perfect channel estimation. The Spatial Processing Unit (SPU), and specifically the Formulation Unit (FU) block receives $\Hbf_i$ and computes the filtering matrix $\Wpi{i}$ (in the figure we drop the subscript P for convenience). FU performs complex conjugate transpose in the case of RMF, and follows steps in Algorithm \ref{algo:IIC} in the case of IIC. $\Wpi{i}$ is then written to the memory.
	During the filtering phase, incoming data vector gets multiplied by $\Wpi{i}$, and its dimensionality reduced from $\Mp \times 1$ to a $\Np \times 1$ ($\Np \ll \Mp$), which is sent to the backplane for further processing.
	
	The SPU is shown in the figure as part of the LDSP, but it is also present in the BDSP architecture. SPU is in charge of data collection, filtering, and distribution. It also performs matrix filtering calculation and its storing. In case of the BDSP architecture, SPU is its main processing element, as in this case FFT and channel estimation is not needed \footnote{Even tough the SPU as a processing unit is identical at each node, data dimensionality may differ from one level to another in the system tree}. The filter can be either $\Wbij{i}{j}$, or $\Wpi{i}$, depending on weather it is part of BDSP or LDSP respectively, and it supports both algorithms. The multiplexers allow to switch between filtering and formulation phase. It is important to notice that same input and output data ports are used during both phases. The dimensionality in both phases is the same. This design decision of using the same SPU architecture throughout the LIS is highly desirable, as it simplifies considerably the design time, verification, and cost of the system. Furthermore, by using the same unit, some or all the backplane nodes may potentially be mapped onto the panels, therefore reducing the number of physical units in the system (in the expense of increasing the workload in panels).
	
	\section{Performance Analysis and Design Trade-offs}
	\label{section:analysis}
	In this section, we analyze the performance and implementation cost of the proposed uplink detection algorithms with the corresponding implementation architecture. More in detail:
	\begin{itemize}
		\item Performance is analyzed based on sum-rate capacity.
		\item Implementation cost in terms of computational complexity, interconnection bandwidth, and processing latency.
	\end{itemize}
	
	The trade-offs between sum-rate capacity and implementation cost is then presented to give high-level design guidelines.

	\subsection{Performance: Optimality and capacity bounds}
	\label{sub:performance}
	
	Closed-form sum-rate expression for multi-panel LIS and IIC algorithm is out of the scope of this work, however we present two upper bounds which provide useful insights. Numerical evaluation of the bounds is shown in next subsection.
	\begin{theorem}
		\label{prop:ub}
		For a certain channel realization $\Hbf$, an upper bound for $C_{\z}$ is given by
		\begin{equation}
		C_{\z} \leq \min \{C_{\mathrm{ub1}}, \Cb\},
		\end{equation}
		where
		\begin{equation}
		\Ca = K \log_2 \left( 1 + \rho\frac{S_{\Np}}{K} \right),
		\end{equation}
		and
		\begin{equation}
		\Cb = \sum_{n=1}^{K}\log_2 (1 + \rho \lambda_{n}),
		\end{equation}
		where $S_{\Np} = \sum_{i=1}^{P}\sum_{n=1}^{\Np}\lambda_{n}^{(i)}$, $\lambda_{n}^{(i)}$ is the $n$-th eigenvalue of $\Hbf_{i}^{H}\Hbf_{i}$, and $\lambda_{n}$ is the $n$-th eigenvalue of $\Hbf^{H}\Hbf$. $P\Np \geq K$ is assumed.	
	\end{theorem}
	\begin{proof}
		$\Cb$ corresponds to the single panel case, then it acts as an upper bound, as always outperforms the multiple-panel case under the same conditions of $P$ and $\Np$. See Appendix-\ref{proof:ub} for proof of $\Ca$.
	\end{proof}

	\subsection{Performance: Experimental results and simulation}

	\begin{figure}
		\footnotesize
		\centering
		\psfrag{LIS}{LIS}
		\psfrag{1.2m}{$1.2m$}
		\psfrag{3m}{$3m$}
		\psfrag{10m}{$10m$}
		\includegraphics[width=\linewidth]{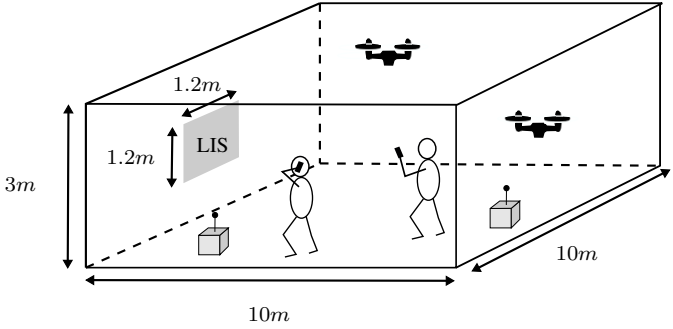}
		\caption{Simulation scenario. A $10m \times 10m \times 3m$ volume, with $1.2m \times 1.2m$ LIS. 64 users uniformly distributed in 3D.}
		\label{fig:LIS_scenario}
		\vspace*{-4mm}
	\end{figure}
	\begin{figure*}[htb]
		\centering
		\subfloat[Low SNR]{
			\includegraphics[width=0.43\linewidth]{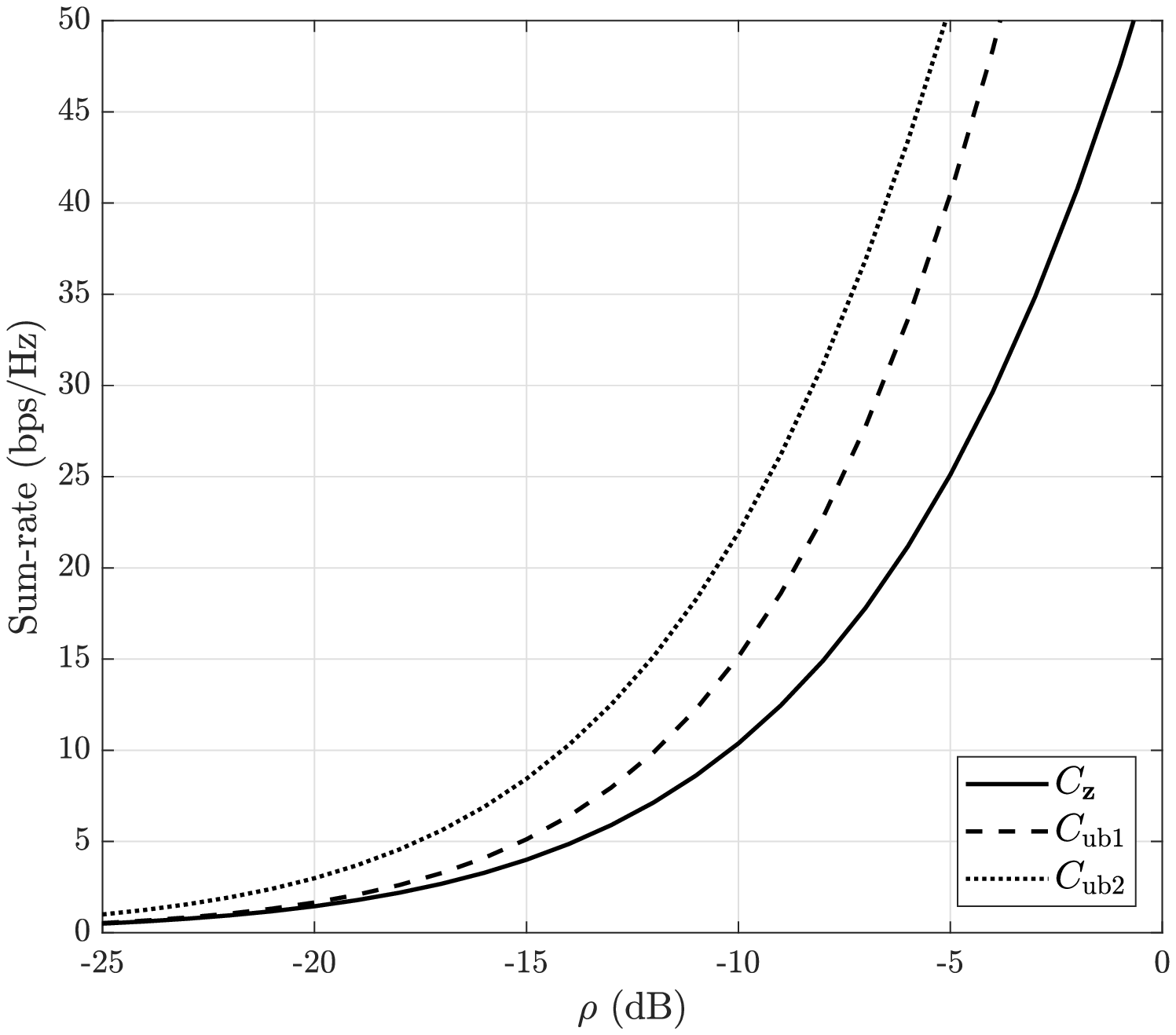}
		}
		\subfloat[High SNR]{
			\includegraphics[width=0.43\linewidth]{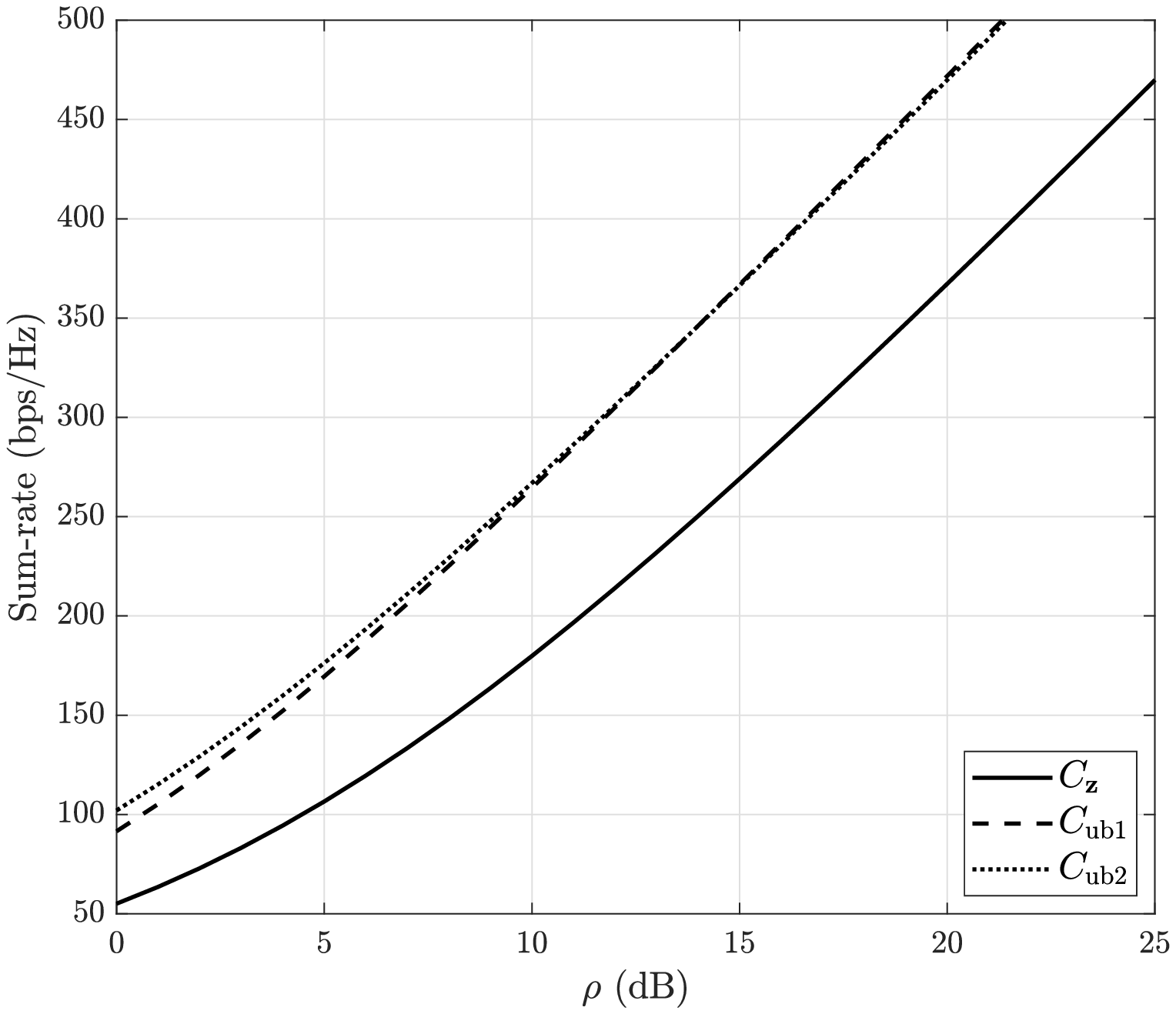}
		}	
		\caption{Average sum-rate capacity at panels output interface vs SNR. Upper bounds in  Proposition \ref{prop:ub} also shown in low and high SNR regimes. $M=1024$, $\Mp=16$, $\Np=2$, and $K=64$.}
		\label{fig:sum_rate_vs_SNR}
		\vspace*{-8mm}
	\end{figure*}
	\begin{figure*}[htb]
		\centering
		\subfloat[RMF]{
			\includegraphics[width=0.43\linewidth]{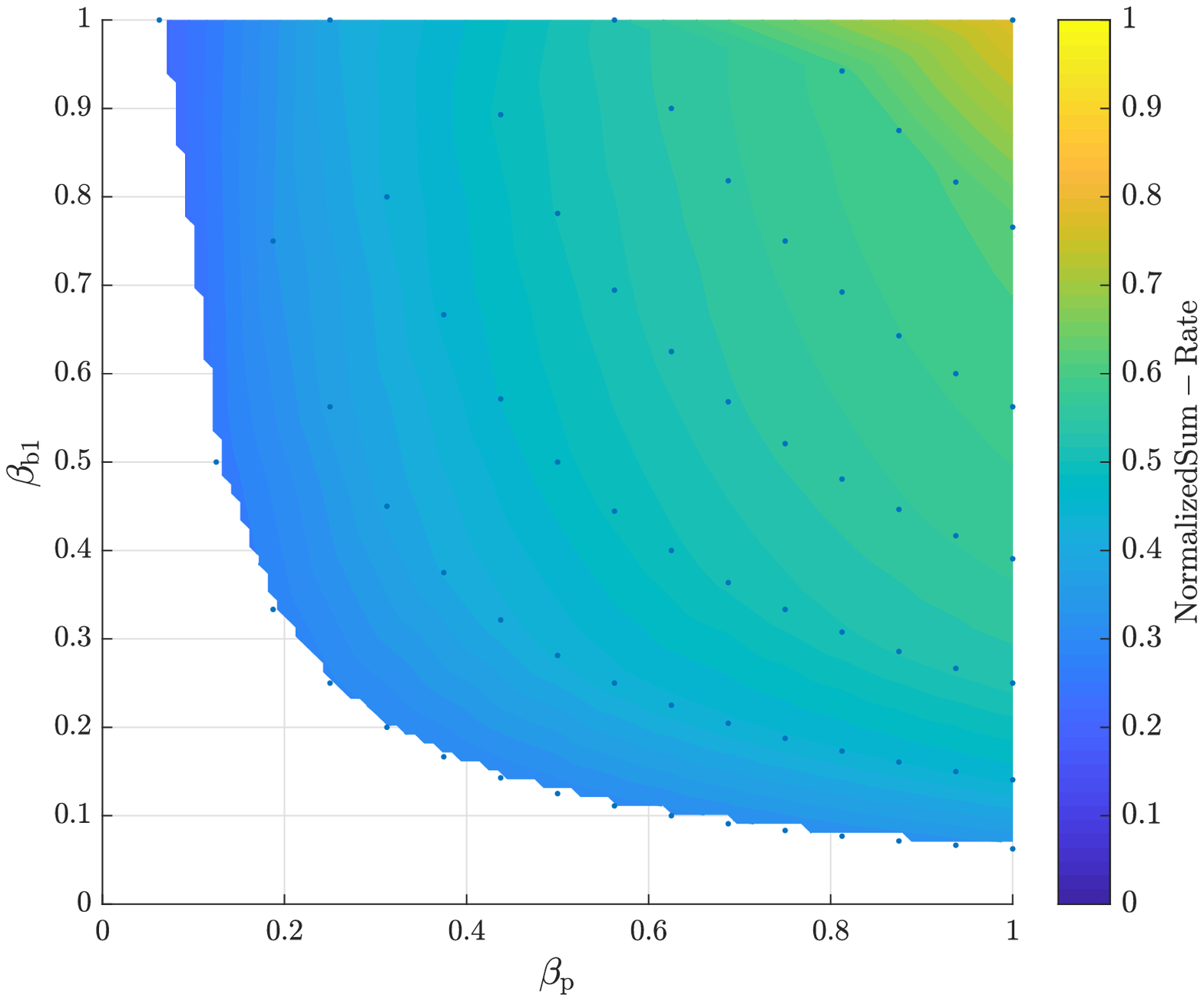}
		}
		\subfloat[IIC]{
			\includegraphics[width=0.43\linewidth]{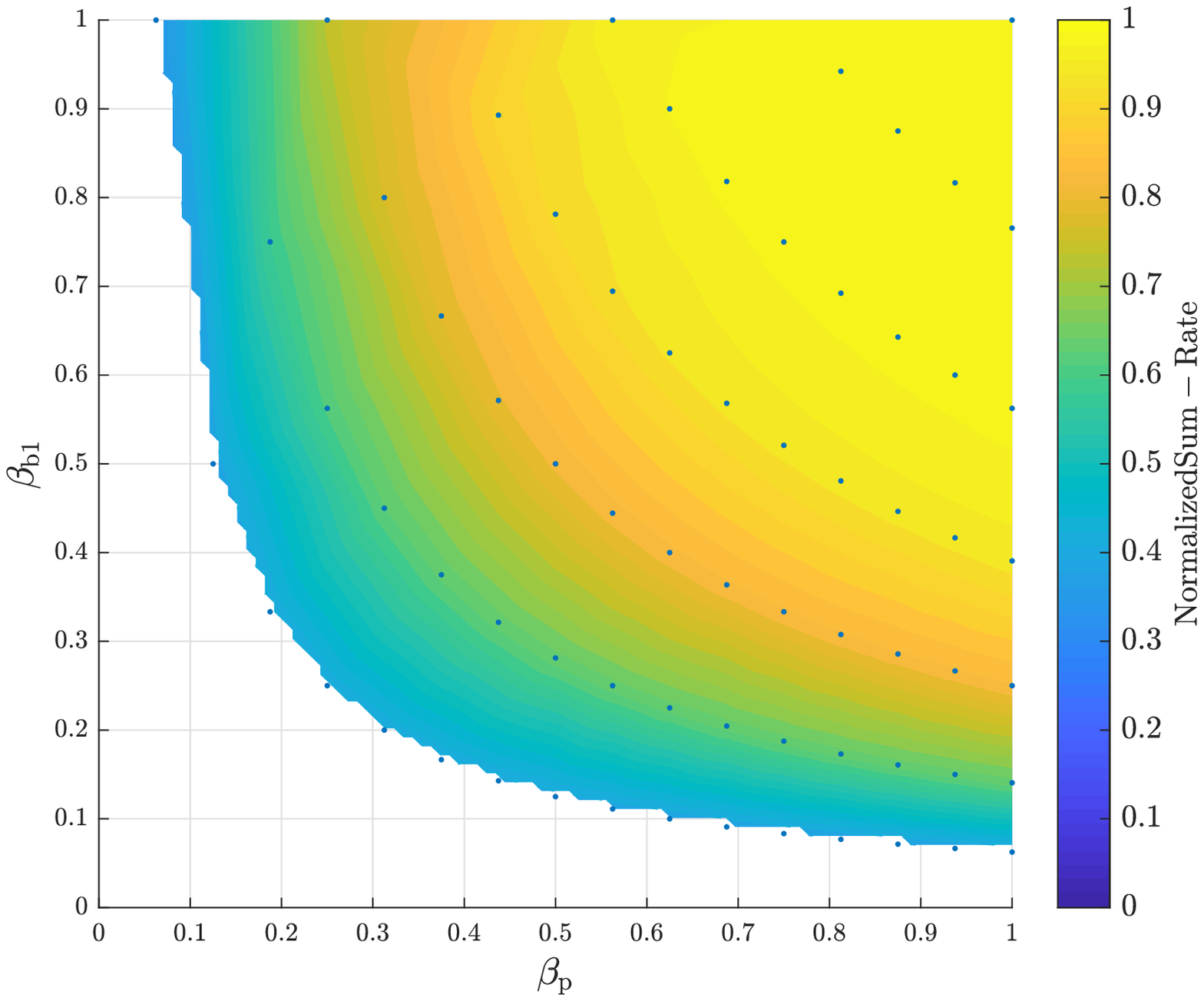}
		}	
		\caption{Sum-rate capacity normalized by channel capacity at CDSP interface for different values of $\beta_{b1}$ vs $\beta_{p}$. $\beta_{b1}=\beta_{b2}=\beta_{b3}$. $M=1024, \Mp=16, K=64$, $\rho=10$. Black dots represent simulated cases. Rest is obtained by linear interpolation.}
		\label{fig:sum_rate_betas}
		\vspace*{-4mm}
	\end{figure*}
	\begin{figure*}[htb]
		\centering
		\subfloat[Sum-rate vs $\Cf$]{
			\includegraphics[width=0.42\linewidth]{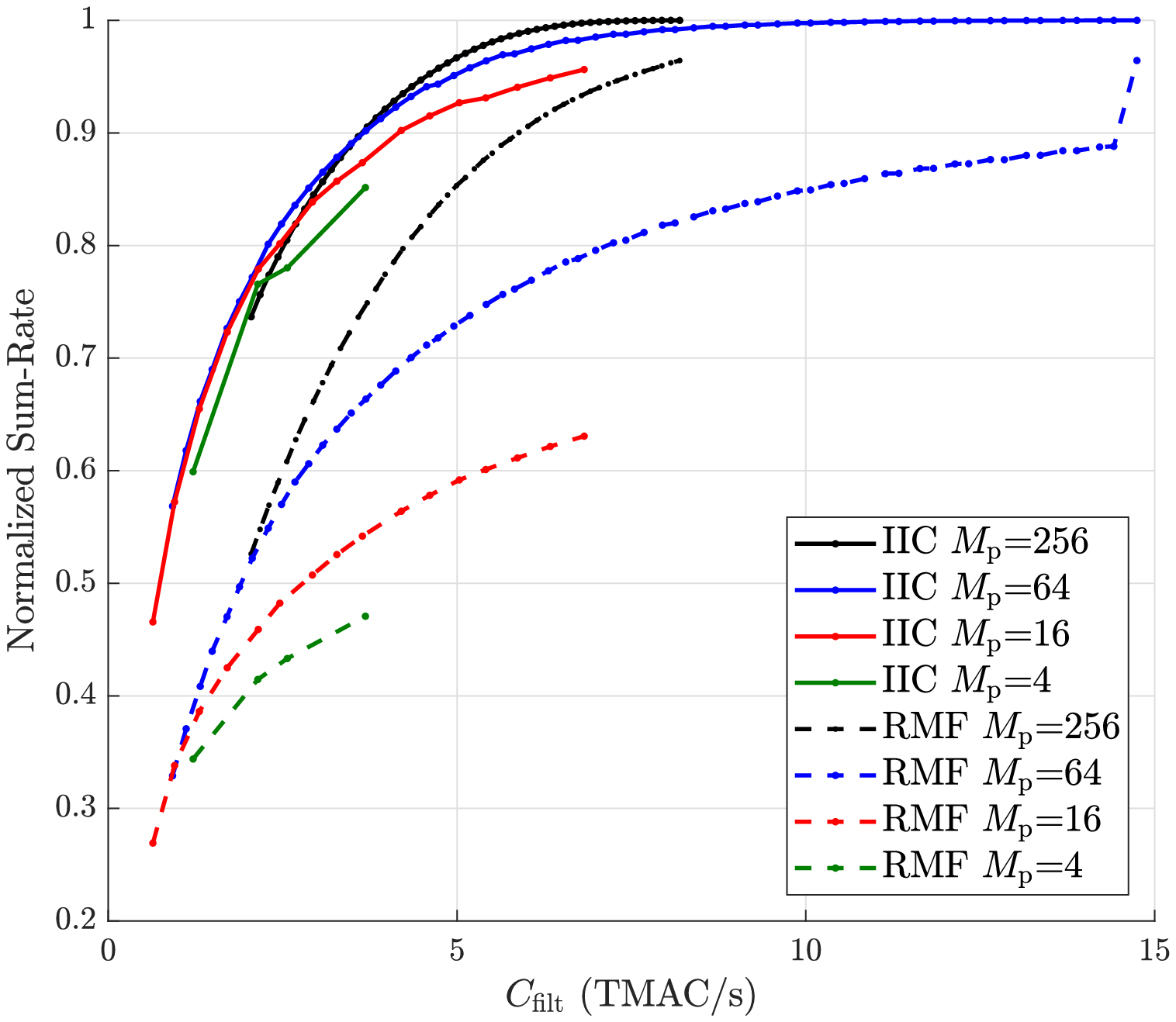}
			\label{fig:SR_vs_C}
		}
		\subfloat[Sum-rate vs $\Req$]{
			\includegraphics[width=0.42\linewidth]{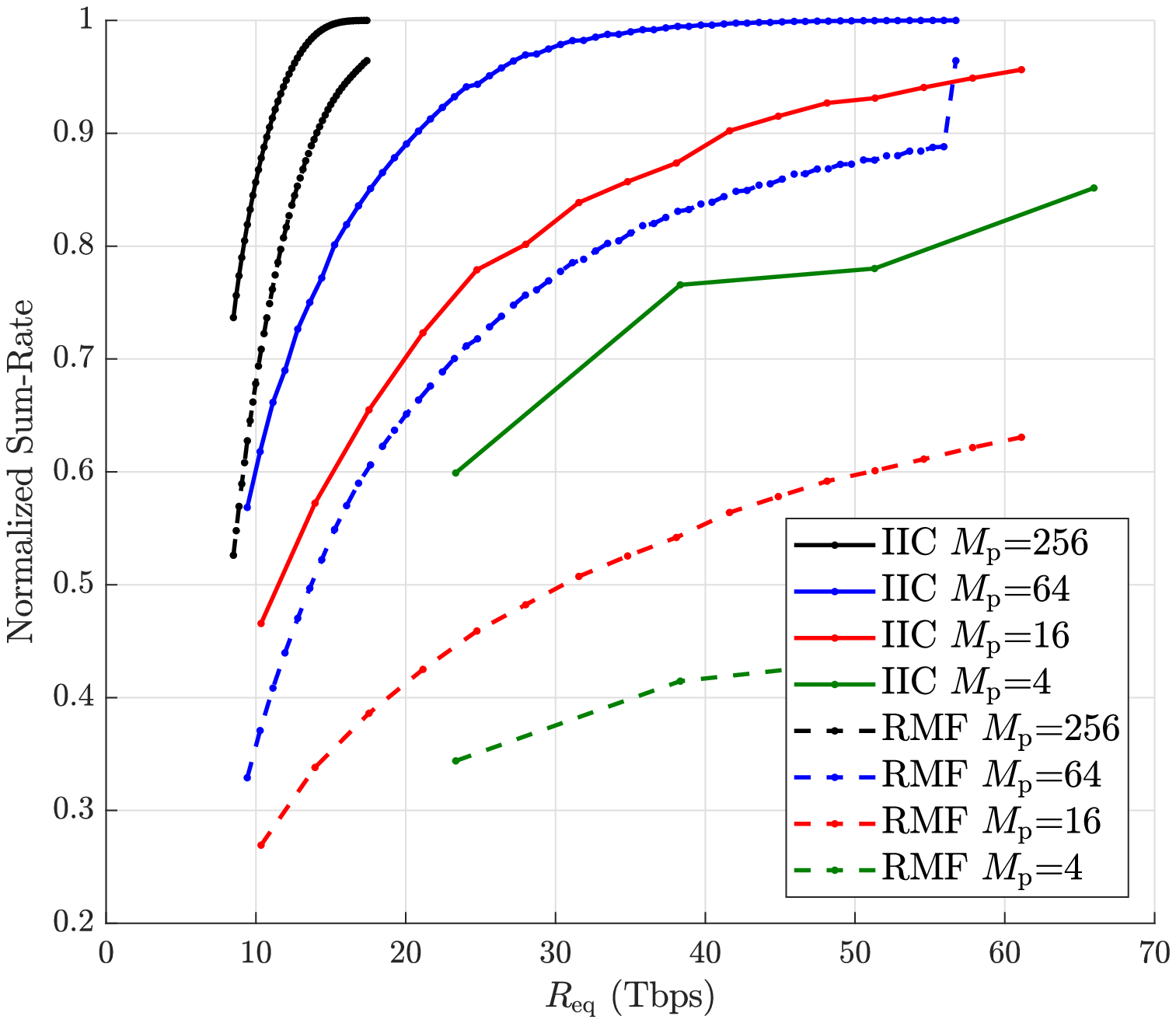}
			\label{fig:SR_vs_R}
		}
		\vspace*{-4mm}
		\\
		\subfloat[Sum-rate vs $\Cf$ for different $M$]{
			\includegraphics[width=0.42\linewidth]{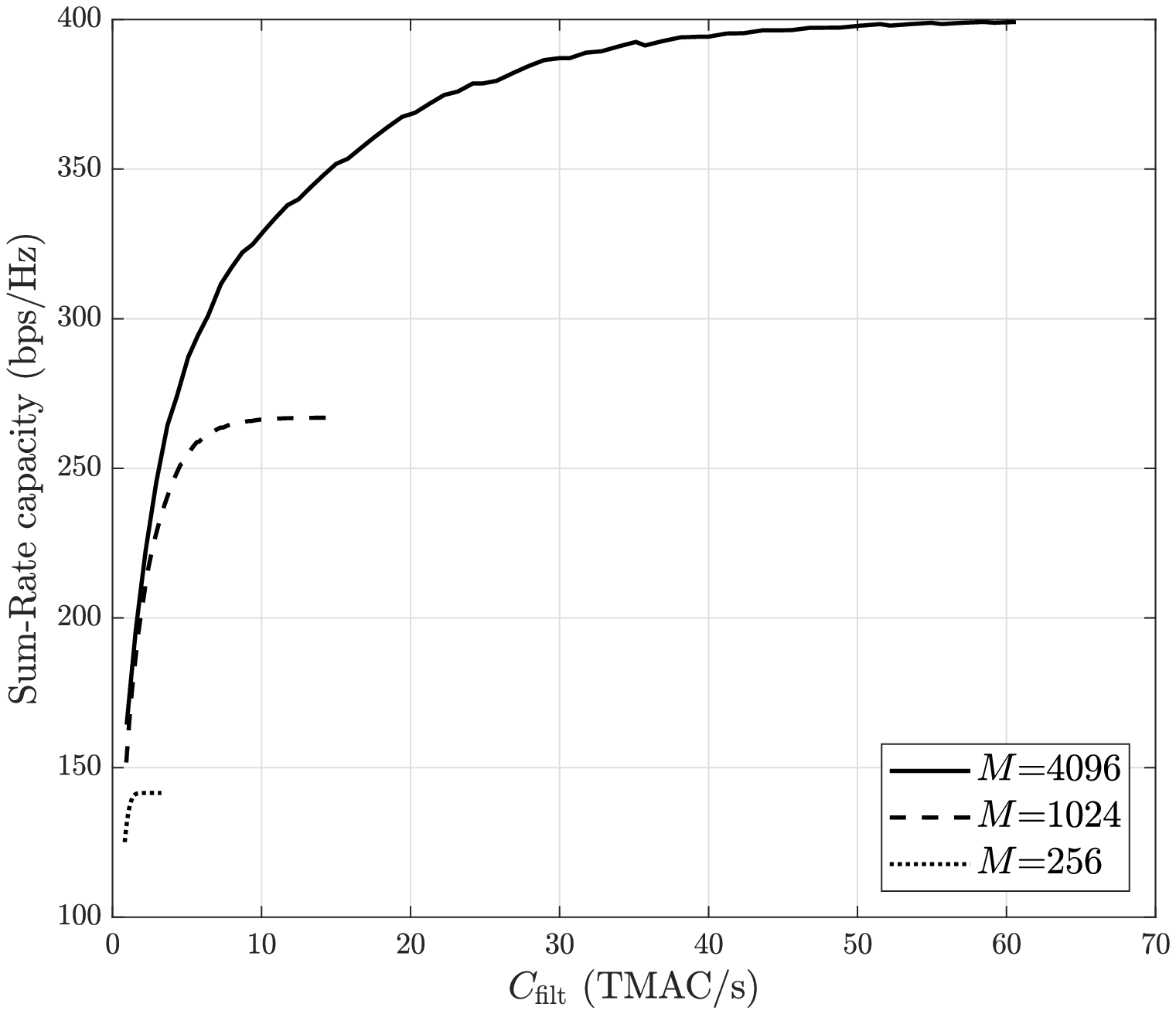}
			\label{fig:SR_vs_C_M}
		}
		\subfloat[Sum-rate vs $\Req$ for different $M$]{
			\includegraphics[width=0.42\linewidth]{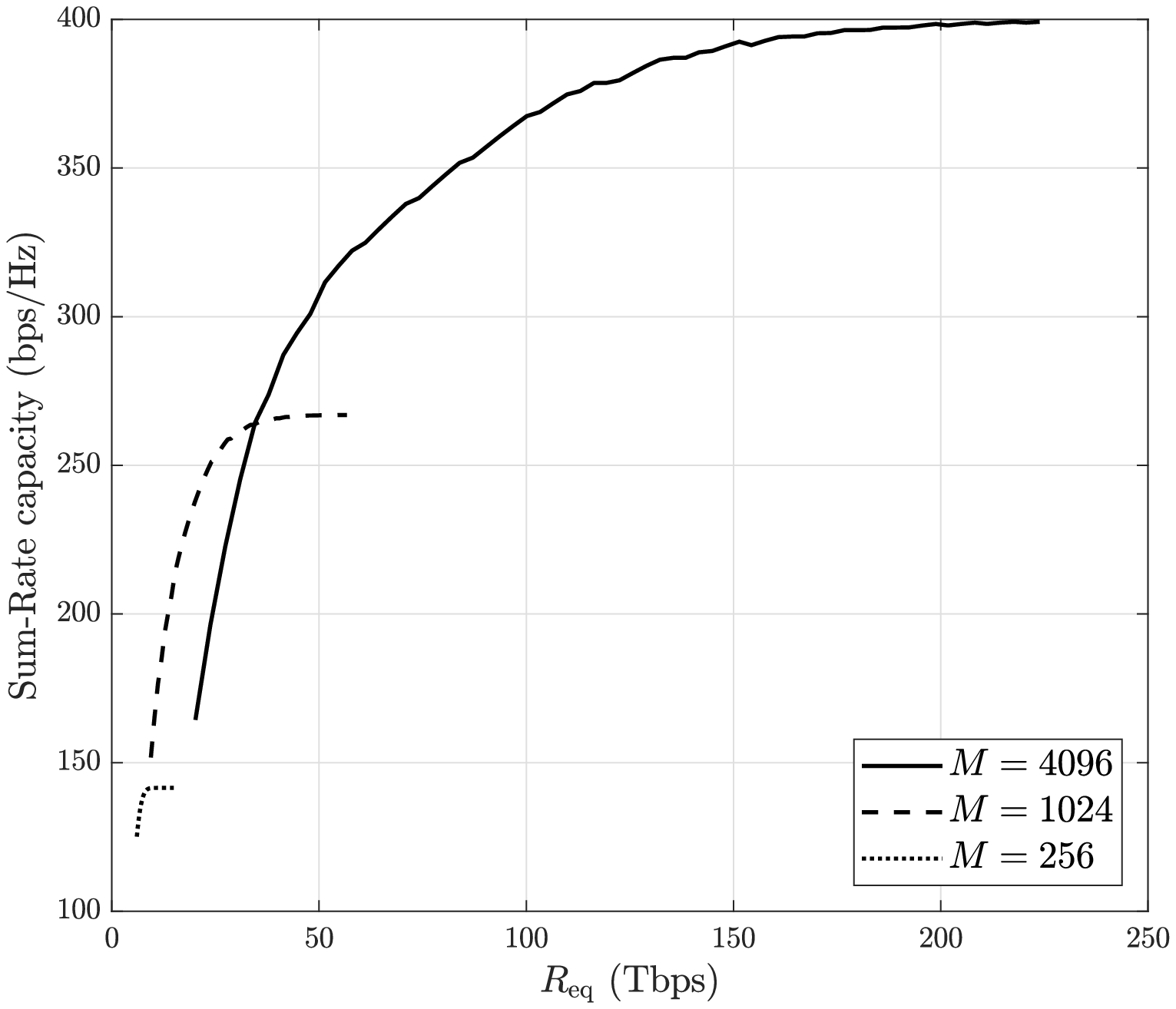}
			\label{fig:SR_vs_R_M}
		}
			
		\caption{Sum-rate capacity normalized by channel capacity at CDSP interface versus computational complexity (\ref{fig:SR_vs_C}) and interconnection data-rate (\ref{fig:SR_vs_R}). In all cases, results for different panel sizes are shown, together with both algorithms. For both cases: $\beta_{\mathrm{b1}}=\beta_{\mathrm{b2}}$. Simulated points represent different $\Np$ values. Sum-rate capacity versus computational complexity (\ref{fig:SR_vs_C_M}), and versus interconnection data-rate for different LIS size (\ref{fig:SR_vs_R_M}). $\Mp=64$, IIC method, and $\rho=10$. $K=64$ in all cases.}
		\label{fig:SR_vs_complexity}
		\vspace*{-4mm}
	\end{figure*}

	The scenario for simulation is shown in Fig. \ref{fig:LIS_scenario}. It consists of 64 users ($K=64$) uniformly distributed in a $10m \times 10m \times 3m$ (depth x width x height) volume in front of a $1.2m \times 1.2m$ (height x width) LIS. Signal bandwidth and carrier frequency are 100MHz and 4GHz respectively. We assume the OFDM-based 5G New Radio (NR) frame structure \cite{5G} and consider uplink processing only.

	To obtain meaningful statistical information we generate 100 channel realizations, by placing the users within the volume following an uniform distribution in the 3 dimensions. For each realization, sum-rate capacity is calculated at different interfaces \footnote{For interfaces we mean panels output, tree nodes outputs, and CDSP input.} of the system, and then averaged across all realizations. The first analysis consists of studying the relation between sum-rate and SNR, and the validity of the bounds in Proposition \ref{prop:ub}. Averaged $C_{\z}$ for $\Np=2$ and different SNR values is shown in Fig. \ref{fig:sum_rate_vs_SNR}, which has been divided in two SNR regions for visual clarity \footnote{The bounds are obtained for the sum-rate at the panels output interface, but are also valid for any other internal interface in the system (such as CDSP input), as sum-rate is lower of equal after each processing in the tree.}. Selection of $\Np=2$ allows us to have enough output panel dimensionality \footnote{$\Np > 2$ also meet this requirement, but at the expense of an increase interconnection bandwidth.}, specifically: $N=128 > K$. Averaged values of the bounds are also shown for comparison. It is clear as $\Ca$ is tight in the low SNR region, while both bounds follow the same slope ($K$) as the sum-rate for high SNR values, with $\sim 5dB$ offset in this case. $\Ca$ is better bound than $\Cb$ in this scenario.

	The sum-rate capacity at CDSP input interface depends on the individual selection of the reduction factor at each node in the system, which leads to a considerable number of possibilities. In order to simplify the analysis and show in a clear form how this individuals selection affects the system performance, let as consider a tree with 3 levels (as in Fig. \ref{fig:model_backplane}) where we constraint the reduction factors as follows: $\betabb=\betabc$, and $\betabc \betabb \betaba \betap = \frac{K}{\Mp}$, where  $\betabi=\frac{\Nbi}{4 N_{\text{bi-1}}}$, $\betap=\frac{\Np}{\Mp}$, and $\betaba=\frac{\Nba}{4\Np}$. By doing so, we ensure there is dimensionality $K$ at the CDSP input for every combination. Therefore, $\beta$ represents the dimensionality reduction at a certain level of the system (all nodes in a certain level are assumed to have same $\beta$ for simplicity\footnote{We foreseen a non-uniform $\beta$ case can be more adequate for scenarios with non-uniform user distribution, which allows to spend resources where it is needed. This is left for further analysis.}), and may take values from 0 (total reduction) to 1 (no reduction). Under this constraint, $\betap$ and $\betaba$ are free to be chosen. Each possible combination provides different sum-rate at CDSP interface, in exchange of different complexity cost \footnote{Computational complexity and interconnection bandwidth.}. Fig. \ref{fig:sum_rate_betas} shows the relation between these two parameters and the normalized sum-rate (value 1 refers to channel capacity measured at antenna interface, and consequently it is the same for both algorithms) for RMF and IIC. It is important to note as multiple ($\betap$, $\betaba$) working points on the same contour level provide the same performance. We verify as a lower reduction (higher $\beta$) leads to higher capacity (but higher interconnection bandwidth), reaching the maximum (or close to it) if no reduction is taking place in the first two levels (point (1,1) in the figure). It is evident as IIC allows higher reduction for same same performance compared to RMF, which translates in lower complexity during filtering, in exchange of higher formulation complexity (computational due to SVD dependency, and interconnection due to the panel-panel local exchange of data).

	\subsection{Computational Complexity}
	\label{sub:complexity}
	We consider number of complex multiplications (MAC) as a metric to measure computational complexity. Our analysis includes both phases: formulation and filtering. In the filtering phase, the operations are the same for RMF and ICC, which consists of applying a liner filter in the panels, of size $\Np \times \Mp$, to the $\Mp \times 1$ input vector, and similar for the BDSP nodes (with different sizes). The total computational complexity for filtering is given by (in MAC/s)		
	\begin{equation}
	\Cf = \underbrace{\fB P \Cf^{(0)}}_{\text{front-end}} + \underbrace{\fB \sum_{n=1}^{L} \Nspu^{(n)} \Cf^{(n)}}_{\text{backplane}},
	\label{eq:Cfilt}
	\end{equation}
	where $\fB$ is the signal bandwidth, $\Cf^{(0)} = \Mp \Np$ is the computational complexity per panel to filter one subcarrier, $\Cf^{(n)} = 4\Nb^{(n-1)} \Nb^{(n)}$ is the corresponding in a node at level $n$, $L$ is the number of levels in the tree, $\Nspu^{(n)}$ is the number of SPUs at level $n$, this is $\Nspu^{(n)} = \frac{P}{4^{n}}$, and $\Nb^{(0)} = \Np$ for notation commodity.
		
	The formulation phase of RMF includes the computation of $\|\h\|^{2}$ for each user. For the IIC algorithm, the steps required for the formulation phase are shown in Algorithm \ref{algo:IIC} for each panel \footnote{For backend there is no exchange of $\Z$ as explained in Section \ref{section:arch}, so the computational complexity is highly reduced.}. This algorithm relies on singular value decomposition (SVD), which we assume is based on 2 steps: Householder bidiagonalization and QR method by Givens rotations. Bidiagonalization is dominant in terms of complexity, so the total complexity of SVD of a $M \times N$ complex matrix can be approximated by $2 M^2N$.
	For step 1 of the Algorithm \ref{algo:IIC}, SVD of a $K \times K$ Gramian matrix $\Z_{i-1}$ is required, with complexity $2K^3$. Step 2 has a complexity of  $(\Mp+1)K^2$, step 3 combined together with 4 require a complexity of  $2\Np d_{0}^2$, where $d_{0}=\max\{K,\Mp\}$, $\Hbf_{eq}^{H}=\Wbfh \Hbf$ consists of $\Np \Mp K$ products, and step 5 of $\Np K^2$. The total computational complexity for IIC is given by (in MACs)\footnote{We assume one channel estimate per PRB, and therefore one filtering matrix calculation per PRB.}
	\begin{equation}
	C_{\text{form,IIC}} = \underbrace{\Nprb P \Cform^{(0)}}_{\text{front-end}} + \underbrace{\Nprb \sum_{n=1}^{L} \Nspu^{(n)} \Cform^{(n)}}_{\text{backplane}}
	\label{eq:Cform_IIC}
	\end{equation}
	where $\Cform^{(0)} = (2K + \Mp + \Np) K^2 + \Np \Mp K + 2\Np d_{0}^2$ is the computational complexity per panel during formulation, while $\Cform^{(n)} = 4\Nb^{(n)}\Nb^{(n-1)}K + 2\Nb^{(n)}d^{2}_{n}$ is the one per node at level $n$,  $d_{n}=\max\{K,4\Nb^{(n-1)}\}$. For RMF we have same expression as \eqref{eq:Cform_IIC} with $\Cform^{(0)}=\Mp K$, and $\Cform^{(n)}=4\Nb^{(n-1)} K$.
	Figures \ref{fig:SR_vs_C} shows normalized sum-rate capacity versus computational complexity during filtering for both algorithms and different panel sizes. We observe as IIC achieves better performance than RMF for same panel size, while large panels are key to harvest most of the capacity, with $\Mp\geq64$ reaching channel capacity in our simulations.
	Figure \ref{fig:SR_vs_C_M} shows sum-rate capacity versus computational complexity during filtering for different LIS size (IIC assumed), and same panel size ($\Mp = 16$). It is very interesting to observe as the same performance (for example 100) can be achieved by both $M=4096$ and $M=1024$, and with same computational complexity. However, their architecture may differ substantially, as the smaller LIS requires higher number of outputs per panel and lower dimensionality reduction, than the larger LIS, where aggressive reduction can be used. In summary, the small LIS ($M=1024$) is harvesting a significant fraction of the available channel capacity, while the larger LIS is only exploiting a very small fraction of it. This presents a very interesting design trade-off.
		
	\subsection{Interconnection bandwidth}
	In this section we analyze the inter-connection bandwidth during filtering phase, covering panel-node and node-node links. This bandwidth is given by (in bps)
	\begin{equation}
	R_{\text{inter}} = \underbrace{2w \fB P \Np}_{\text{front-end}} + \underbrace{2w \fB \sum_{n=1}^{L} \Nspu^{(n)} \Nb^{(n)}}_{\text{backplane}},
	\nonumber
	\end{equation}
	where $w$ is the bit-width of the SPU input/output (real and imaginary parts).	In our analysis we also consider the movement of data happening internally at panels/nodes level, which covers the data transfer between the inputs ports to the SPU, for processing, and from it to the output ports. We name this transfer data-rate as $\textit{intra-connection data-rate}$ or $R_{\text{intra}}$ \footnote{We are aware that $\Rintra$ does not include all internal data-rate in a real system, as this is highly dependent on the specific implementation, internal topology, and type of processing unit employed in the panel. However, the spirit of this work is to provide a general analysis and first order approximation of the complexity required, applicable to all possible implementations, instead of being attached to an specific hardware implementation, and provide exact analysis numbers.}, and
	\begin{equation}
	\begin{aligned}
	R_{\text{intra}} &= \underbrace{2w \fB P (\Mp + \Np)}_{\text{front-end}} \\
	&+\underbrace{2w \fB \sum_{n=1}^{L} \Nspu^{(n)} (4\Nb^{(n-1)}+\Nb^{(n)}) }_{\text{backplane}}.
	\end{aligned}
	\nonumber
	\end{equation}

	In order to take both magnitudes into consideration in our analysis, we define the relative cost $\alpha$, as $\alpha \triangleq \text{cost}(\Rintra)/\text{cost}(\Rinter)$, and the cost equivalent inter-connection data-rate $\Req$ as: $\Req \triangleq \Rinter + \alpha \Rintra$.
	In this analysis, we take power/data-rate as cost magnitude. If we assume serial link (serdes) technology for intra-connection, and Ethernet for inter-connection, then we obtain a power consumption of $1.29-24.8$mW/Gps, and $40$mW/Gbps respectively according to different sources \cite{serdes1,serdes2,eth1,eth2}. The serdes power range is very wide, so as an example we take 4mW/Gbps as reference, that gives $\alpha \sim \frac{1}{10}$ \footnote{These numbers are dependent on the technology used, however, the method still holds.}.
	
	Figure \ref{fig:SR_vs_R} shows normalized sum-rate capacity versus equivalent inter-connection bandwidth during filtering for both algorithms and different panel sizes. We observe as IIC achieves better performance than RMF for same panel size, and large panels are capable to harvest most of the channel capacity in our simulations. It is relevant to point out that small panels require more total interconnection data-rate than large panels, however this is more distributed among panels and nodes, reducing considerably the bottlenecks.
	Fig. \ref{fig:SR_vs_R_M} shows sum-rate capacity versus interconnection bandwidth during filtering for different LIS size (IIC assumed). Similar conclusions can be extracted compared to Fig. \ref{fig:SR_vs_C_M}.
	
	\subsection{Processing Latency}
	The processing latency represents the time between when the estimated channel of a subcarrier is available at panels and  when the data of that subcarrier is filtered and available at the CDSP input for detection. The latency can be expressed as $L_{\text{tot}} = L_{\text{form}} + L_{\text{filt}}$, where $L_{\text{form}}$ is the formulation latency, and $L_{\text{filt}}$ is the latency for data filtering. More specifically, $L_{\text{form}} = L^{\text{proc}}_{\text{form}} + (\nP - 1) L^{\text{com}}_{\text{local}} + (L+1) L^{\text{com}}_{\text{global}}$, where $L^{\text{proc}}_{\text{form}}$ is the time needed to calculate the filter coefficients, $L^{\text{com}}_{\text{local}}$ refers to panel-to-panel communication latency (only in IIC), and $L^{\text{com}}_{\text{global}}$ refers to panel-to-node, and node-to-node link communication latency. $\nP$ is the number of panels involved ($\nP = 1$ in RMF and $P$ in IIC for the worst case) \footnote{Depending on the users distribution we may not need to go through all panels ($\nP < P$) with the subsequent benefits. We leave both items for future work.}. For filtering latency we have: $L_{\text{filt}} = L^{\text{proc}}_{\text{filt}} + (L+1) L^{\text{com}}_{\text{global}}$, which accounts for filtering in panels and nodes, and communication latency. We assume the IIC formulation is done sequentially along all panels (worst case) using local connections, and then across nodes in the tree.
	
	The latency for processing highly depends on the hardware architecture used to implement the algorithms. Here we assume highly optimized accelerators (e.g., ASIC) are used that the available data parallelism ($\Nparal$) can be explored using $\Nproc$ processing units ($\Nproc<\Nparal$), i.e., the $\Nproc$ PEs will take $\Nparal/\Nproc$ clock cycles to iteratively process $\Nparal$ parallel operations. Moreover, the channel matrix (of the subcarrier that is being processed) is cached in register files (the latency for memory access is hidden). The main component of $L^{\text{proc}}_{\text{form}}$ is the time needed to perform SVD which is implemented by Householder bidiagonalization followed by QR method based on Givens rotations. The processing of each column and row can be done in parallel, while sequential processing is needed between columns and rows due to the data dependency.
	
	With these assumptions, the total processing latency in formulation phase is $L^{\text{proc}}_{\text{form}} = \frac{\widetilde{C}_{\text{form}} T_{\text{CLK}}}{\Nproc}$, where $\widetilde{C}_{\text{form}} = \nP \Cform^{(0)} + \sum_{n=1}^{L}\Cform^{(n)}$. The first term in $\widetilde{C}_{\text{form}}$ represents the serial processing in the front-end, and the second term represents the computational complexity of one branch of the tree. $\Cform^{(0)}$ and $\Cform^{(n)}$ are defined after \eqref{eq:Cform_IIC}. $T_{\text{CLK}}$ is the clock period, and we assume that one complex multiplication and accumulation (MAC) can be done within one clock cycle. In case of filtering, processing latency is given by $L^{\text{proc}}_{\text{filt}} = \frac{\widetilde{C}_{\text{filt}} T_{\text{CLK}}}{\Nproc}$,	where $\widetilde{C}_{\text{filt}} = \sum_{n=0}^{L}\Cf^{(n)}$ is the computational complexity corresponding to a path between a panel and the CDSP, and $\Cf^{(n)}$ is defined after \eqref{eq:Cfilt}.
	
	\begin{table}[t!]
		\centering
		\begin{tabular}{c|c|c|c|c|c|c|c} 
			\hline
			$\textbf{Method}$ & $\Cform$ & $\Cf$ &  $\Rinter$ & $\Rintra$ & $L_{\text{form}}$ & $L_{\text{filt}}$ \\
			\hline
			\hline
			IIC &  3.1 & 2.3 & 1.0 & 5.4 & 110.2 & 1.0\\
			\hline
			RMF & 0.02 & 2.3 & 1.0 & 5.4 & 1.2 & 1.0\\ 
			\hline
		\end{tabular}
		\caption{Values of total complexity for a LIS: $M=1024$, $\Mp=64$, $K=50$, $\betap=1/4,\betaba=1/2$. $w = 12$ bits. $\Nprb = 275$. $\fB=100$MHz. Units are as follows: $\Cform$ [GMAC], $\Cf$ [TMAC/s], $\Rinter$ [Tb/s], $\Rintra$ [Tb/s], $L$ [$\mu$s].}
		\label{table:Comp_total}
		\vspace*{-4mm}
	\end{table}

	\subsection{Case study and discussion}
	Performance has been analyzed, together with computational complexity, inter-connection data-rate, and processing latency. General expressions for these different magnitudes have been presented based on general system parameters, such as number of users, number of antennas, number of panels, and signal bandwidth among others; what makes it easy to particularize for concrete implementations. Nevertheless, based on the trade-off analysis shown in Fig. \ref{fig:SR_vs_C} and Fig. \ref{fig:SR_vs_R}, we can see $\Mp=64$ as an attractive option, as it provides higher capacity than $\Mp=16$ for same computational complexity and interconnection data-rate, while it is able to reach channel capacity in our analysis scenario. It is also of a reasonable size in case we want to distribute the LIS in a certain area. On top of that, its physical dimensions makes it easy to handle and mount ($30cm \times 30cm$ at 4GHz). For this panel size we present numerical values of the analyzed complexity in Table \ref{table:Comp_total}. The following parameter values are assumed: $\nP=P=16$, $T_{\text{CLK}} = 1ns$, $N_{\text{paral}} = 100$, $L^{\text{com}}_{\text{local}}=100ns$ (serdes technology assumed \cite{serdes1,serdes2}), and $L^{\text{com}}_{\text{global}}=300ns$ (ethernet assumed \cite{eth1,eth2,eth_ti}). Assuming 12 subcarriers per PRB, the subcarrier spacing in our example is: $\frac{\fB}{12\Nprb}=30$KHz, and the OFDM symbol duration is therefore $\approx 33 \mu s$.
	
	The benefits of the distributed architecture are evident in terms of interconnection data-rate reduction. If we look at the CDSP input interface, the reduction is easily obtained as: $\frac{M}{K} \sim 20$x. Of course, this is in exchange of a performance loss due to dimensionality reduction, but as we have explained before the system is fully configurable, offering a rich performance-complexity trade-off. It is important to consider that even tough computational complexity and inter-connection data rates numbers may seem large, they are distributed among all processing units in the LIS. This LIS contains 21 SPUs (panels + backplane nodes).	
	
	Regarding latency, $L_{\text{form,RMF}}$ and $L_{\text{filt}}$ values seem reasonable for NR frame structure. We observe as $L_{\text{form,IIC}}$ shows much higher value due to the higher computational complexity required in this method (in this example equivalent to 3 OFDM symbols). For a certain LIS system this latency is sensitive to the $\beta$ used in panels and nodes (which translates into complexity cost), and $K$ (system capacity). Therefore we can foresee a trade-off between these system parameters and how often filters are updated in panels and nodes. It is important to remark that we analyzed latency from a worst case point of view, where all panels in the LIS are serially connected and jointly contribute to formulation. In reality we do not think this is the best approach as this may only be helpful in cases with very high density of users with dominant interference over noise. We foresee groups of panels performing serial processing within, but parallel among groups, reducing considerably the formulation latency.
	
	We are aware that depending on the implementation latency may be different (selection of memory system, hardware, interconnection), and here we provide high level analysis assuming we use dedicated accelerators without any overhead.

	\section{Conclusions}
	\label{section:conclusions}
	In this article we have presented distributed uplink processing algorithms and the corresponding hardware architecture for efficient implementation of large intelligent surfaces (LIS). The proposed processing structure consists of local panel processing units to reduce incoming data dimensionality without losing much information, and hierarchical backplane network with distributed processing-combining units to support flexible and efficient data aggregation. We have systematically analyzed the system capacity and implementation cost with different design parameters, and provided design guidelines for the implementation of LIS.

	\appendix
	
	\subsection{Proof of Proposition \ref{prop:ub}}
	\label{proof:ub}
	\begin{proof}
	The sum-rate capacity with the multi-panel architecture is given by
	\begin{equation}
	C = \log_2 |\IK + \rho \A|,
	\end{equation}
	where $\A = \sum_{i=1}^{P} \Hbfh_{i} \Q_{i} \Q^{H}_{i} \Hbf_{i}$. For a certain channel realization, the maximum capacity is achieved if all eigenvalues of $\A$ are equal, this is: $\lambda_n = \overline{\lambda}, 1 \leq n \leq K$ \footnote{Note that we assume $rank(\A)=K$, and $P\Np \geq K$.}. In that case, the capacity would be: $\Ca = K \log_2 (1+\rho \overline{\lambda})$. Now, let as find the maximum value for $\overline{\lambda}$ as follows
	\begin{equation}
	\begin{aligned}
	\overline{\lambda} &= \max_{\{\Q_i\}} \frac{1}{K} \Tr \{A\}
	= \frac{1}{K} \sum_{i=1}^{P} \max_{\Q_i} \Tr \{\Hbfh_{i} \Q_{i} \Q^{H}_{i} \Hbf_{i}\}\\
	&= \frac{1}{K} \sum_{i=1}^{P} \sum_{n=1}^{\Np} \lambda_{n}^{(i)},\\
	\end{aligned}
	\end{equation}
	and then: $C \leq \Ca$, so the proposition is proven. \footnote{We remark this bound may not be attained in practice, as it needs a favorable set of $\{\Hbf_i\}$, or in other way, there may not be such $\{\Q_i\}$ that provides uniform eigenvalues in $\A$.}
	\end{proof}
		
	\subsection{Proof of solution to local optimization in Algorithm \ref{algo:IIC_i}}
	\label{proof:sol_IIC_i}
	\begin{proof}
	We drop the panel index for simplicity. The objective function to maximize is
	\begin{equation}
	\begin{aligned}
	|\rho \Hbf^{H} \Q \Q^{H} \Hbf + \mathbf{Z}|
	&= |\Z| |\IK + \rho \Z^{-1/2} \Hbfh \Q\Q^{H} \Hbf \Z^{-1/2} |\\
	&= |\Z| |\I_{\Np} + \rho \Q^{H} \Hbf \Z^{-1} \Hbfh \Q|\\
	&= |\Z| |\Q^{H}(\rho\Hbf \Z^{-1} \Hbfh + \I_{\Mp})\Q|.
	\end{aligned}
	\nonumber
	\end{equation}

	$|\Z|$ does not depend on $\Q$, therefore the solution to our problem is the same as the solution of the maximization of the second determinant, which consists of the ordered eigenvectors (in descent order of corresponding eigenvalue) of the matrix: $\Hbf \Z^{-1} \Hbfh$.
	\end{proof}

	\subsection{Proof of white filtered noise}
	\label{proof:white_noise}
	As an example, the filtered noise due to the first first four panels and the node connected to them is denoted as $n_{1}^{(1)}$ and obtained as: $n_{1}^{(1)} = \Wbhij{1}{1} \Wphi{1-4} \mathbf{n}_{1-4}$, where $\Wpi{1-4}$ is the combined filtering matrix of the first four panels and it is defined as $\Wpi{1-4} = \diag(\Wpi{1},\Wpi{2},\Wpi{3},\cdots, \Wpi{4})$, and $\mathbf{n}_{1-4}$ is the aggregated input noise vector corresponding to the first four panels and it is defined as $\mathbf{n}_{1-4} = [\mathbf{n}_{1}, \mathbf{n}_{2}, \mathbf{n}_{3}, \mathbf{n}_{4}]^{T}$. The covariance is given by
	\begin{equation}
	\begin{split}
	\E \left\lbrace n_{1}^{(1)} n_{1}^{(1)H} \right\rbrace &= \Wbhij{1}{1} \Wphi{1-4} \E \{\mathbf{n}_{1-4} \mathbf{n}_{1-4}^{H}\} \Wpi{1-4}\Wbij{1}{1}\\
	&= \Wbhij{1}{1} \Wphi{1-4} \I_{4\Mp} \Wpi{1-4} \Wbij{1}{1}\\
	&= \Wbhij{1}{1} \I_{4\Np} \Wbij{1}{1} = \I_{\Nb^{(1)}}
	\end{split}
	\nonumber
	\end{equation}

	\bibliographystyle{IEEEtran}
	\bibliography{IEEEabrv,LIS}
	
\end{document}

%% file: packages.tex
\usepackage{cite}
\usepackage{amsmath,amsthm,amssymb,amsfonts}
\usepackage{algorithmic}
\usepackage{subfig}
\usepackage{graphicx}

\usepackage{textcomp}
\usepackage{xcolor}

\usepackage{epstopdf}

\usepackage{setspace}

\usepackage{psfrag}

\usepackage[linesnumbered,lined,boxed,commentsnumbered, ruled]{algorithm2e} %required for algorithms pseudocode

\usepackage{array,multirow}
\usepackage{threeparttable}

%% file: commands.tex
\newtheorem{theorem}{Proposition}

\newcommand{\Tr}{\mathrm{tr}}

\newcommand{\diag}{\mathrm{diag}}

\newcommand{\I}{\mathbf{I}}
\newcommand{\IK}{\mathbf{I}_{K}}
\newcommand{\A}{\mathbf{A}}

\newcommand{\E}{\mathbb{E}}

\newcommand{\Q}{\mathbf{Q}}

\newcommand{\U}{\mathbf{U}}
\newcommand{\V}{\mathbf{V}}
\newcommand{\Z}{\mathbf{Z}}
\newcommand{\Wbf}{\mathbf{W}}
\newcommand{\Wbfh}{\mathbf{W}^{H}}
\newcommand{\Htilde}{\widetilde{\U}_{H}}

\newcommand{\Hbf}{\mathbf{H}}
\newcommand{\Heq}{\mathbf{H}_{\text{eq}}}
\newcommand{\Hbfh}{\mathbf{H}^{H}}
\newcommand{\Sbf}{\mathbf{\Sigma}}

\newcommand{\Wp}{\mathbf{W}_{\mathrm{P}}}
\newcommand{\Wph}{\mathbf{W}^{H}_{\mathrm{P}}}
\newcommand{\Wpi}[1]{\mathbf{W}_{\mathrm{P},#1}}
\newcommand{\Wphi}[1]{\mathbf{W}^{H}_{\mathrm{P},#1}}
\newcommand{\Wb}{\mathbf{W}_{\mathrm{B}}}
\newcommand{\Wbij}[2]{\mathbf{W}_{\mathrm{B},#1}^{(#2)}}
\newcommand{\Wbhij}[2]{\mathbf{W}_{\mathrm{B},#1}^{(#2)H}}
	
\newcommand{\h}{\mathbf{h}}

\newcommand{\y}{\mathbf{y}}
\newcommand{\x}{\mathbf{x}}
\newcommand{\z}{\mathbf{z}}
\newcommand{\s}{\mathbf{s}}

\newcommand{\hatu}{\hat{\mathbf{u}}}
\newcommand{\hatn}{\hat{\mathbf{n}}}

\newcommand*{\argmax}{arg\,max}

\newcommand{\Nprb}{N_{\text{PRB}}}

\newcommand{\fB}{f_{\text{B}}}

\newcommand{\Mp}{M_{\text{p}}}
\newcommand{\Np}{N_{\text{p}}}
\newcommand{\Nb}{N_{\text{b}}}

\newcommand{\Cf}{C_{\text{filt}}}
\newcommand{\Cform}{C_{\text{form}}}
\newcommand{\Nspu}{N_{\text{SPU}}}
\newcommand{\Rinter}{R_{\text{inter}}}
\newcommand{\Rintra}{R_{\text{intra}}}
\newcommand{\Req}{R_{\text{eq}}}
\newcommand{\betap}{\beta_{\text{p}}}
\newcommand{\betaba}{\beta_{\text{b1}}}
\newcommand{\betabb}{\beta_{\text{b2}}}
\newcommand{\betabc}{\beta_{\text{b3}}}
\newcommand{\betabi}{\beta_{\text{b}i}}
\newcommand{\Nba}{N_{\text{b1}}}

\newcommand{\Nbi}{N_{\text{b}i}}

\newcommand{\nP}{n_{\text{P}}}

\newcommand{\Ca}{C_{\mathrm{ub1}}}
\newcommand{\Cb}{C_{\mathrm{ub2}}}

\newcommand{\Nparal}{N_{\text{paral}}}
\newcommand{\Nproc}{N_{\text{proc}}}